\documentclass[11pt]{article}
\usepackage[T2A]{fontenc}
\usepackage{ alphalph,etoolbox }
\usepackage{amsthm, amsmath, amssymb, amsbsy, amscd, amsfonts, latexsym, euscript,epsf, bm}
\newtheorem{thm}{Theorem}[section]
\usepackage{bbold, bm}
\usepackage{epic,eepic}
\usepackage{texdraw}
\usepackage[dvips]{color}
\usepackage{color}
\usepackage{ dsfont }
\usepackage{graphicx}
\usepackage{exscale,relsize}
\usepackage{authblk}
\usepackage{url}
\usepackage{enumerate}
\usepackage{graphicx}
\newtheorem{proposition}[thm]{Proposition}

\newtheorem{lemma}[thm]{Lemma}

\newtheorem{remark}[thm]{Remark}

\newtheorem{theorem}[thm]{Theorem}

\DeclareMathOperator{\End}{End}
\allowdisplaybreaks

\title{Noncommutative solutions to Zamolodchikov's tetrahedron equation and matrix six-factorisation problems}
\date{}

\author[1]{S. Konstantinou-Rizos\thanks{skonstantin84@gmail.com}}
\affil[1]{Centre of Integrable Systems, P.G. Demidov Yaroslavl State University, Yaroslavl, Russia}

\patchcmd{\subequations}{\alph{equation}}{\alphalph{\value{equation}}}{}{}

\begin{document}

\maketitle

\begin{abstract}
It is known that the local Yang--Baxter equation is a generator of potential solutions to Zamolodchikov's tetrahedron equation. In this paper, we show under which additional conditions the solutions to the local Yang--Baxter equation are tetrahedron maps, namely solutions to the set-theoretical tetrahedron equation. This is exceptionally useful when one wants to prove that noncommutative maps satisfy the Zamolodchikov's tetrahedron equation. We construct new noncommutative maps and we prove that they possess the tetrahedron property. Moreover, by employing Darboux transformations with noncommutative variables, we derive noncommutative tetrahedron maps. In particular, we derive a noncommutative nonlinear Schr\"odinger type of tetrahedron map which can be restricted to a noncommutative version of Sergeev's map on invariant leaves. We prove that these maps are tetrahedron maps.
\end{abstract}

\bigskip

\hspace{.2cm} \textbf{PACS numbers:} 02.30.Ik, 02.90.+p, 03.65.Fd.

\hspace{.2cm} \textbf{Mathematics Subject Classification 2020:} 35Q55, 16T25.

\hspace{.2cm} \textbf{Keywords:} Functional tetrahedron equation, Zamolodchikov's equation, noncommutative tetrahe-

\hspace{2.4cm} dron maps, local Yang--Baxter equation, matrix six-factorisation problem, Darboux

\hspace{2.4cm} transformations, noncommutative NLS type tetrahedron maps, noncommutative 

\hspace{2.4cm} Sergeev's map.

\section{Introduction}\label{intro}
The Zamolodchikov's tetrahedron equation is a higher-dimensional analogue of the Yang--Baxter equation, one of the most fundamental equations of mathematical physics, and it was first introduced by Zamolodchikov \cite{Zamolodchikov, Zamolodchikov-2}. The first solutions to the tetrahedron equation were conjectured by Zamolodchikov himself, based on numerical works, but were proved to be solutions by Baxter \cite{Baxter-1983} who also calculated the free energy of the associated solvable three-dimensional model in the limit of an infinite lattice \cite{Baxter-1986}.

Tetrahedron maps, namely solutions to the Zamolodchikov's functional tetrahedron equation, are of great significance in the the theory of integrable systems since they are strictly related to integrable three-dimensional lattice equations (see, e.g., \cite{Bazhanov-Mangazeev-Sergeev, Doliwa-Kashaev, Kassotakis-Tetrahedron} and the references therein) which also discretise nonlinear integrable PDEs, and at the same time have very interesting algebro-geometric properties (see, e.g., \cite{Bazhanov-Sergeev, Bazhanov-Mangazeev-Sergeev, IKKRP, Sergei-Sokor, Kapranov-Voevodsky}). On the other hand, noncommutative versions or extensions of integrable systems have been a growing field over the past few decades, with many applications in mathematical physics, and have been in the centre of interest for many scientists (indicatively we refer to \cite{Bobenko-Suris, Dimakis-Hoissen, Dimakis-Hoissen-2015, Doliwa-Noumi, Kupershmidt, Nijhoff-Capel, Nimmo, Talalaev}). Therefore, there is a natural need to study the noncommutative solutions to the tetrahedron equation. Indeed, there are recent results in the literature on the study of noncommmutative tethrahedron maps \cite{Doliwa-Kashaev, Sergei-Sokor}.

One important relation between the tetrahedron equation and the Yang--Baxter equation is that the solutions to the local Yang--Baxter equation are possible solutions to the tetrahedron equation \cite{Maillet-Nijhoff, Kashaev-Sergeev}. In fact, a map which is derived by substitution of a square matrix to the local Yang--Baxter equation may satisfy the tetrahedron equation \cite{Maillet-Nijhoff}. In the commutative case, the proof that a map is a tetrahedron map is a matter of straightforward substitution to the tetrahedron equation. However, if one deals with noncommutative variables, since the substitution to the tetrahedron equation involves several compositions between nonlinear maps, the proof of the tetrahedron property is a difficult task and depends on the form of the map itself.

In this paper, we show which additional matrix condition must be satisfied in order for a solution of the local Yang--Baxter equation to be a tetrahedron map. This involves the consideration of a matrix six-factorisation problem, and it is motivated by similar results obtained by Kouloukas and Papageorgiou \cite{Kouloukas-Papageorgiou} for verifying whether a solution to a matrix refactorisation problem is a Yang--Baxter map. This matrix six-factorisation property is exceptionally useful when one works with maps with noncommutative variables, since it can be used to prove whether a map is a tetrahedron map without actually using the map, but just its matrix generator. We construct new examples of tetrahedron maps with matrix generators related to a gauge transformation for the lattice modified KdV (mKdV) equation \cite{Frank-Walker} and also to one-dimensional relativistic elastic collision of two particles \cite{Kouloukas}. Furthermore, we show how can one construct noncommutative tetrahedron maps using Darboux transformations. As an illutrative example, we employ a Darboux transformation related to the noncommutative coupled NLS system which gives rise to a noncommutative six-dimensional tetrahedron map which can be restricted on inviariant leaves to a noncommutative version of Sergeev's map \cite{Sergeev}. We employ the matrix six-factorisation condition to prove the tetrahedron property of all the maps derived in this paper. 

The paper is organised as follows. 

In the next section, we fix the notation we are using throughout the text, and we provide the basic definitions which are needed for the text to be self-contained. In particular, we give the definition of a tetrahedron map and a parametric tetrahedron map, and also we explain the relation between the latter and the local Yang--Baxter equation, namely we define the Lax representation for tetrahedron maps.

In section \ref{Matrix-six-factorisation_problem}, we explain what additional conditions must the solutions to the local Yang--Baxter equation satisfy in order to be tetrahedron maps. Specifically, we prove that a solution of the local Yang--Baxter equation is a tetrahedron map if a certain matrix six-factorisation condition implies the trivial solution. Furthermore, we demonstrate how this is useful when we deal with noncommutative maps using a noncommutative Hirota map \cite{Doliwa-Kashaev} as an illustrative example. Finally, we construct new correspondences, which define novel, noncommutative tetrahedron maps, generated by matrices related to a gauge  transformation for the lattice mKdV equation \cite{Frank-Walker} and also to a one-dimensional relativistic elastic collision of two particles \cite{Kouloukas}. Using these examples, we demonstrate that the matrix six-factorisation condition can be also used for correspondences which satisfy the local Yang--Baxter equation.

Section \ref{NLS-noncomm_tetrahedron_maps} deals with the construction of noncommutative tetrahedron maps. The tetrahedron property is proven for the maps of this section using the matrix six-factorisation condition presented in section \ref{Matrix-six-factorisation_problem}. In particular, we employ a noncommutative Darboux transformation of NLS type in order to construct a correspondence satisfying the local Yang--Baxter equation which defines a noncommutative tetrahedron map. We show that the latter map can be restricted to a noncommutative version of Sergeev's map on invariant leaves.

In section \ref{conclusions}, we close with some concluding remarks and ideas for possible extensions of our results.

\section{Preliminaries}\label{prelim}
In this section, we explain the relation between the solutions to the local Yang--Baxter equation and the solutions to the functional tetrahedron equation.

\subsection{Notation}
Throughout the text:
\begin{itemize}
    \item By $\mathcal{X}$ we denote an arbitrary set, whereas by Latin italic letters (i.e. $x, y, u, v$ etc.) the elements of $\mathcal{X}$, with an exception of the `spectral parameter' which is denoted by the Greek letter $\lambda$. Moreover, by $\End(\mathcal{X})$ we denote any map $\mathcal{X}\rightarrow\mathcal{X}$.
    \item By $\mathfrak{R}$ we denote a noncommutative division ring, and its elements are denoted by bold italic Latin letters (i.e. $\bm{x}, \bm{y}, \bm{u}$ etc.). That is, $\mathfrak{R}$ is an associative algebra with multiplicative identity $1$ where commutativity with respect to mutliplication is not assumed ($\bm{x}\bm{y}\neq \bm{y}\bm{x}$), and every nonzero element $\bm{x}$ has an inverse $\bm{x}^{-1}$, i.e. $\bm{x}\bm{x}^{-1}=\bm{x}^{-1}\bm{x}=1$.
    \item The centre of a division ring will be denoted by $Z(\mathfrak{R})=\{a\in\mathfrak{R}:\forall\bm{x}\in\mathfrak{R},a\bm{x}=\bm{x}a\}$.
    \item Matrices will be denoted by capital Roman straight letters (i.e. ${\rm A}, {\rm B}, {\rm C}$) etc. Additionally, matrix operators are denoted by capital Gothic letters (for instance, $\mathfrak{L}=D_x+\rm{U}$).
\end{itemize}

\subsection{Zamolodchikov's functional tetrahedron VS local Yang--Baxter equation}
A map $T\in\End(\mathcal{X}^3)$, namely
\begin{equation}\label{Tetrahedron_map}
 T:(x,y,z)\mapsto (u(x,y,z),v(x,y,z),w(x,y,z)),
\end{equation}
is called a \textit{tetrahedron map} if it satisfies the \textit{functional tetrahedron} (or Zamolodchikov's tetrahedron) equation
\begin{equation}\label{Tetrahedron-eq}
    T^{123}\circ T^{145} \circ T^{246}\circ T^{356}=T^{356}\circ T^{246}\circ T^{145}\circ T^{123}.
\end{equation}
Functions $T^{ijk}\in\End(\mathcal{X}^6)$, $i,j=1,2,3,~i\neq j$, in \eqref{Tetrahedron-eq} are maps that act as map $T$ on the $ijk$ terms of the Cartesian product $\mathcal{X}^6$ and trivially on the others. For instance,
$$
T^{246}(x,y,z,r,s,t)=(x,u(y,r,t),z,v(y,r,t),s,w(y,r,t)).
$$

A tetrahedron map can be represented on the cube as in Figure \ref{Tet_map}.

\begin{figure}[ht]
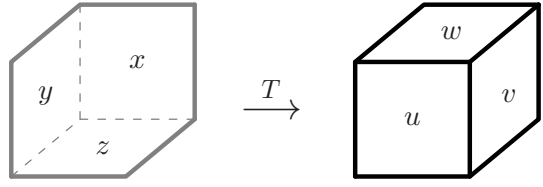

\centering
\centertexdraw{ 
\setunitscale .6
\move(1.6 1.5) \setgray 0.5 \linewd 0.04 \lvec (0.6 1.5) \lvec(0 1) \lvec(0 0)
\move (0 0) \linewd 0.04 \lvec (1 0) \lvec(1.6 .5) \linewd 0.01  \lpatt(0.067 0.09) \lvec(0.6 .5) \lvec(0 0) 
\lpatt() \move(1.6 .5)\linewd 0.04 \lvec(1.6 1.5) 
 \linewd 0.01 \lpatt(0.067 0.09)  \move(0.6 1.5)\lvec(0.6 .5) \lpatt()
\move(3 0)\setgray 0 \linewd 0.04 \lvec(3 1)\lvec(3.6 1.5)  \lvec(4.6 1.5)\setgray 0 \linewd 0.04
\move (3 0) \linewd 0.04 \lvec (4 0) \lvec(4.6 .5) \lvec(4.6 1.5) \linewd 0.04 \lvec(4 1) \lvec(4 0) 
\move(3 1)\lvec(4 1)

\htext (2 .5) {{\Large $\overset{T}{\longrightarrow}$}}
\textref h:C v:C \htext(.8 .25){$z$}
\textref h:C v:C \htext(.3 .7){$y$}
\textref h:C v:C \htext(1.1 1){$x$}

\textref h:C v:C \htext(3.8 1.25){ $w$}
\textref h:C v:C \htext(4.3 .7){ $v$}
\textref h:C v:C \htext(3.5 .5){$u$}

}
\caption{Tetrahedron map. Schematic representation.}\label{Tet_map}
\end{figure}

Furthermore, if we assign the complex parameters $a$, $b$ and $c$ to the variables $x$, $y$ and $z$, respectively, we define a map $T\in\End[(\mathcal{X}\times\mathbb{C})^3]$, namely $T:((x,a),(y,b),(z,c))\mapsto ((u(x,y,z),a),(v(x,y,z),b),(w(x,y,z),c))$ which we denote for simplicity as
\begin{equation}\label{Par-Tetrahedron_map}
 T_{a,b,c}:(x,y,z)\mapsto (u_{a,b,c}(x,y,z),v_{a,b,c}(x,y,z),w_{a,b,c}(x,y,z)).
\end{equation}
Map \eqref{Par-Tetrahedron_map} is called a \textit{parametric tetrahedron map} if it satisfies the \textit{parametric functional tetrahedron equation}
\begin{equation}\label{Par-Tetrahedron-eq}
    T^{123}_{a,b,c}\circ T^{145}_{a,d,e} \circ T^{246}_{b,d,f}\circ T^{356}_{c,e,f}=T^{356}_{c,e,f}\circ T^{246}_{b,d,f}\circ T^{145}_{a,d,e}\circ T^{123}_{a,b,c}.
\end{equation}

Now, let ${\rm L}={\rm L}(x,k)$ be a matrix depending on a variable $x\in\mathcal{X}$ and a parameter $k\in\mathbb{C}$ of the form
\begin{equation}\label{matrix-L}
   {\rm L}(x,k)= \begin{pmatrix} 
a(x,k) & b(x,k)\\ 
c(x,k) & d(x,k)
\end{pmatrix},
\end{equation}
where its entries $a, b, c$ and $d$ are scalar functions of $x$ and $k$. Let ${\rm L}^3_{ij}$, $i,j=1,2, 3$, $i\neq j$, be the $3\times 3$ extensions of matrix \eqref{matrix-L},  defined by
{\small
\begin{equation}\label{Lij-mat}
   {\rm L}^3_{12}=\begin{pmatrix} 
 a(x,k) &  b(x,k) & 0\\ 
c(x,k) &  d(x,k) & 0\\
0 & 0 & 1
\end{pmatrix},\quad
 {\rm L}^3_{13}= \begin{pmatrix} 
 a(x,k) & 0 & b(x,k)\\ 
0 & 1 & 0\\
c(x,k) & 0 & d(x,k)
\end{pmatrix}, \quad
 {\rm L}^3_{23}=\begin{pmatrix} 
   1 & 0 & 0 \\
0 & a(x,k) & b(x,k)\\ 
0 & c(x,k) & d(x,k)
\end{pmatrix},
\end{equation}
}
where ${\rm L}^3_{ij}={\rm L}^3_{ij}(x,k)$, $i,j=1,2,3$.

The following matrix trifactorisation problem
\begin{equation}\label{Lax-Tetra}
    {\rm L}^3_{12}(u,a){\rm L}^3_{13}(v,b){\rm L}^3_{23}(w,c)= {\rm L}^3_{23}(z,c){\rm L}^3_{13}(y,b){\rm L}^3_{12}(x,a),
\end{equation}
where matrices $L^3_{ij}$ are defined as in \eqref{Lij-mat}, is the Maillet--Nijhoff equation \cite{Nijhoff} in Korepanov's form, which appears in the literature as the \textit{local Yang--Baxter} equation.

Now, if a map of the form \eqref{Par-Tetrahedron_map} satisfies the local Yang--Baxter equation \eqref{Lax-Tetra}, then this map is a possible tetrahedron map. If, indeed, the above matrix trifactorisation problem defines a tetrahedron map, we will call equation \eqref{Lax-Tetra} its \textit{Lax representation}. Equation \eqref{Lax-Tetra} was used by Kashaev, Korepanov and Sergeev to classify the solutions to the Zamolodchikov's tetrahedron equation. In this paper, we consider the case where $a(x,k), b(x,k), c(x,k)$ and $d(x,k)$ in \eqref{Lij-mat} are scalar functions, however Korepanov studied equation \eqref{Lax-Tetra} the case where $a(x,k), b(x,k), c(x,k)$ and $d(x,k)$ in \eqref{Lij-mat} are matrices \cite{Korepanov}. 

In the next section, we will show under which conditions the solutions of the local Yang--Baxter equation are solutions to the Zamolodchikov's tetrahedron equation. 

\section{Matrix six-factorisation problem and tetrahedron maps}\label{Matrix-six-factorisation_problem}
Proving that a map satisfies the tetrahedron equation by straightforward substitution to the equation is usually a difficult task when we deal with noncommutative variables; it involves several compositions of maps with noncommutative variables and their inverses. Here, following the work of Kouloukas and Papageorgiou \cite{Kouloukas-Papageorgiou} for the case Yang--Baxter maps, we find the additional matrix conditions that the solutions of the local Yang--Baxter equation must satisfy in order to be tetrahedron maps.

Let $L^4_{ij}$, $i,j=1,\ldots 4$, $i\neq j$, be the $4\times 4$ generalisations of matrix ${\rm L}\equiv {\rm L}(x,k)=\begin{pmatrix} a(x,k) & b(x,k)\\ c(x,k) & d(x,k)\end{pmatrix}$, namely
\begin{align}\label{Lij4}
        &{\rm L}^4_{12}=\begin{pmatrix}a(x,k) & b(x,k) & 0 & 0 \\ c(x,k) & d(x,k) & 0 & 0 \\ 0 & 0 & 1 & 0 \\ 0 & 0 & 0 & 1\end{pmatrix}, \quad
    {\rm L}^4_{13}=\begin{pmatrix}a(x,k) & 0 & b(x,k) & 0 \\ 0 & 1 & 0 & 0 \\ c(x,k) & 0 & d(x,k) & 0 \\ 0 & 0 & 0 & 1\end{pmatrix},\quad
        {\rm L}^4_{23}=\begin{pmatrix}1 & 0 & 0 & 0 \\ 0 & a(x,k) & b(x,k) & 0 \\ 0 & c(x,k) & d(x,k) & 0 \\ 0 & 0 & 0 & 1\end{pmatrix}\nonumber\\
   & {\rm L}^4_{14}=\begin{pmatrix}a(x,k) & 0 & 0 & b(x,k) \\ 0 & 1 & 0 & 0 \\ 0 & 0 & 1 & 0 \\ c(x,k) & 0 & 0 & d(x,k)\end{pmatrix}, \quad
        {\rm L}^4_{24}=\begin{pmatrix}1 & 0 & 0 & 0 \\ 0 & a(x,k) & 0 & b(x,k) \\ 0 & 0 & 1 & 0 \\ 0 & c(x,k) & 0 & d(x,k)\end{pmatrix}, \quad
    {\rm L}^4_{34}=\begin{pmatrix}1 & 0 & 0 & 0 \\ 0 & 1 & 0 & 0 \\ 0 & 0 & a(x,k) & b(x,k) \\ 0 & 0 & c(x,k) & d(x,k)\end{pmatrix}.
\end{align}
Throughout the text the entries of the above matrices $a(x,k)$, $b(x,k)$, $c(x,k)$ and $d(x,k)$ are scalar functions of a variable $x$ and a parameter $k$. However, the following results can be generalised for the case where $a(x,k)$, $b(x,k)$, $c(x,k)$ and $d(x,k)$ are matrices.


\begin{lemma}\label{LijLji}
If matrices ${\rm L}^4_{ij}$ and ${\rm L}^4_{kl}$, where $i,j,k,l=1,\ldots, 4$, $i<j, k<l$ and $j\neq k$, are defined as in \eqref{Lij4}, then they commute. That is: ${\rm L}^4_{ij}{\rm L}^4_{jk}={\rm L}^4_{jk}{\rm L}^4_{ij}$.
\end{lemma}

\begin{lemma}\label{3-eq-equiv}
Let matrices ${\rm L}_{ij}$ be defined as in \eqref{Lij4}. Then, the following equations are equivalent:{\small
\begin{align*}
    &{\rm L}^4_{12}(u,a){\rm L}^4_{13}(v,b){\rm L}^4_{23}(w,c)={\rm L}^4_{23}(z,c){\rm L}^4_{13}(y,b){\rm L}^4_{12}(x,a),\\
    &{\rm L}^4_{12}(u,a){\rm L}^4_{14}(v,b){\rm L}^4_{24}(w,c)={\rm L}^4_{24}(z,c){\rm L}^4_{14}(y,b){\rm L}^4_{12}(x,a),\\
    &{\rm L}^4_{23}(u,a){\rm L}^4_{24}(v,b){\rm L}^4_{34}(w,c)={\rm L}^4_{34}(z,c){\rm L}^4_{24}(y,b){\rm L}^4_{23}(x,a).
\end{align*}
}
\end{lemma}
\begin{proof}
By straightforward substitution of matrix ${\rm L}(x,k)=\begin{pmatrix} a(x,k) & b(x,k)\\ c(x,k) & d(x,k)\end{pmatrix}$ to the above equations, one can show that all the above imply exactly the same system of polynomial equations for $u$, $v$ and $w$, $x$, $y$ and $z$.
\end{proof}

Now, we act on the elements $\left(x,y,z,r,s,t\right)$ with the left and right part of the tetrahedron equation, as indicated in Figure \ref{Tet_eq-rep}. Specifically, using the right-hand side of the tetrahedron equation, we have
\begin{align*}
  T^{123}\left(x,y,z,r,s,t\right)&=\left(\tilde{x},\tilde{y},\tilde{z},r,s,t\right),\\
 T^{145}\circ T^{123}\left(x,y,z,r,s,t\right)&=\left(\tilde{\tilde{x}},\tilde{y},\tilde{z},\tilde{r},\tilde{s},t\right),\\
 T^{246}\circ  T^{145}\circ T^{123}\left(x,y,z,r,s,t\right)&=\left(\tilde{\tilde{x}},\tilde{\tilde{y}},\tilde{z},\tilde{\tilde{r}},\tilde{s},\tilde{t}\right),\\
 T^{356}\circ T^{246}\circ  T^{145}\circ T^{123}\left(x,y,z,r,s,t\right)&=\left(\tilde{\tilde{x}},\tilde{\tilde{y}},\tilde{\tilde{z}},\tilde{\tilde{r}},\tilde{\tilde{s}},\tilde{\tilde{t}}\right),
\end{align*}
whereas the left-hand side of the tetrahedron equation implies
\begin{align*}
  T^{356}\left(x,y,z,r,s,t\right)&=\left(x,y,\hat{z},r,\hat{s},\hat{t}\right),\\
  T^{246}\circ T^{356}\left(x,y,z,r,s,t\right)&=\left(x,\hat{y},\hat{z},\hat{r},\hat{s},\hat{\hat{t}}\right),\\
 T^{145}\circ T^{246}\circ T^{356}\left(x,y,z,r,s,t\right)&=\left(\hat{x},\hat{y},\hat{z},\hat{\hat{r}},\hat{\hat{s}},\hat{\hat{t}}\right),\\
  T^{123}\circ T^{145}\circ T^{246}\circ T^{356}\left(x,y,z,r,s,t\right)&=\left(\hat{\hat{x}},\hat{\hat{y}},\hat{\hat{z}},\hat{\hat{r}},\hat{\hat{s}},\hat{\hat{t}}\right),
\end{align*}
as in Figure \ref{Tet_eq-rep}. That is, if $\hat{\hat{x}}=\tilde{\tilde{x}}$, $\hat{\hat{y}}=\tilde{\tilde{y}}$, $\hat{\hat{z}}=\tilde{\tilde{z}}$, $\hat{\hat{r}}=\tilde{\tilde{r}}$, $\hat{\hat{s}}=\tilde{\tilde{s}}$ and $\hat{\hat{t}}=\tilde{\tilde{t}}$, then map \eqref{Tetrahedron_map} satisfies the tetrahedron equation and vice versa.

\begin{figure}[ht]
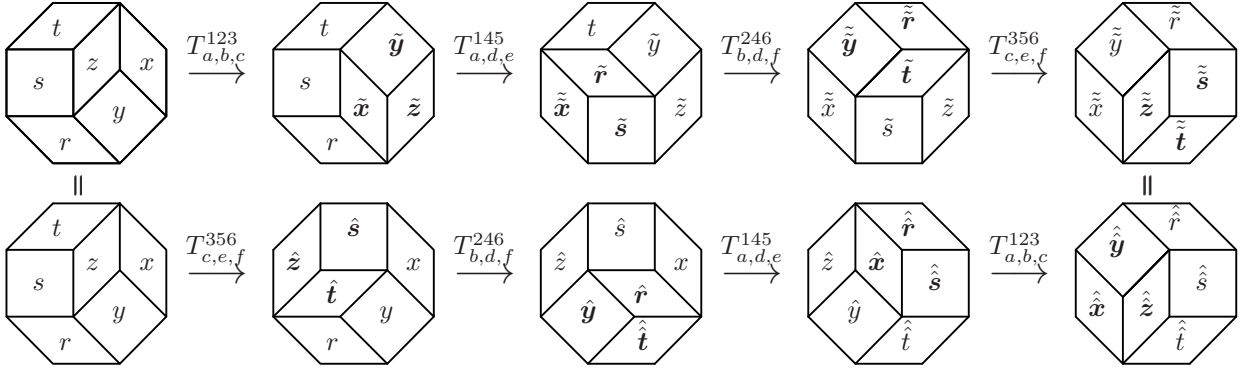

\centering
\centertexdraw{ 
\setunitscale .35

\move (0 0) \lvec (0 -1)\lvec(0.7 -1.7)\lvec(0.7 -0.7)\lvec(0 0)
\move (0 0) \lvec (-1 0)
\move (0 0)\lvec(-0.7 -0.7)\lvec(-0.7 -1.7)\lvec(0 -1)
\move (-0.7 -0.7)\lvec(-1.7 -0.7)\lvec(-1 0)
\move(-0.7 -1.7)\lvec(0 -2.4)\lvec(0.7 -1.7)
\move(-0.7 -1.7)\lvec(-1.7 -1.7)\lvec(-1.7 -0.7)
\move(0 -2.4)\lvec(-1 -2.4)\lvec(-1.7 -1.7)

\htext (-0.5 -1){{\small $z$}}
\htext (0.3 -1){{\small $x$}}
\htext (-0.1 -1.8){{\small $y$}}
\htext (-.9 -2.2){{\small $r$}}
\htext (-1.3 -1.3){{\small $s$}}
\htext (-1 -.5){{\small $t$}}

\htext (1 -1.2) {{\Large $\overset{T^{123}_{a,b,c}}{\longrightarrow}$}}

\move (4 0) \lvec(4.7 -0.7)\lvec(4.7 -1.7)
\move (4 0) \lvec (3 0)
\move (4 0)\lvec(3.3 -0.7)\lvec(3.3 -1.7)
\move(3.3 -0.7)\lvec(4 -1.4)\lvec(4.7 -0.7)
\move(4 -1.4)\lvec(4 -2.4)
\move (3.3 -0.7)\lvec(2.3 -0.7)\lvec(3 0)
\move(3.3 -1.7)\lvec(4 -2.4)\lvec(4.7 -1.7)
\move(3.3 -1.7)\lvec(2.3 -1.7)\lvec(2.3 -0.7)
\move(4 -2.4)\lvec(3 -2.4)\lvec(2.3 -1.7)

\htext (3.9 -0.8){{\small { $\tilde{\bm{y}}$}}}
\htext (3.5 -1.7){{\small {$\tilde{\bm{x}}$}}}
\htext (4.3 -1.7){{\small {$\tilde{\bm{z}}$}}}
\htext (3.1 -2.2){{\small $r$}}
\htext (2.7 -1.3){{\small $s$}}
\htext (3 -.5){{\small $t$}}

\htext (5 -1.2) {{\Large $\overset{T^{145}_{a,d,e}}{\longrightarrow}$}}

\move (8 0) \lvec(8.7 -0.7)\lvec(8.7 -1.7)
\move (8 0) \lvec (7 0)
\move (8 0)\lvec(7.3 -0.7)
\move(7.3 -0.7)\lvec(8 -1.4)\lvec(8.7 -0.7)
\move(8 -1.4)\lvec(7 -1.4)\lvec(7 -2.4)
\move(7 -1.4)\lvec(6.3 -0.7)
\move(8 -1.4)\lvec(8 -2.4)
\move (7.3 -0.7)\lvec(6.3 -0.7)\lvec(7 0)
\move(8 -2.4)\lvec(8.7 -1.7)
\move(6.3 -1.7)\lvec(6.3 -0.7)
\move(8 -2.4)\lvec(7 -2.4)\lvec(6.3 -1.7)

\htext (7.9 -0.8){{\small $\tilde{y}$}}
\htext (6.5 -1.7){{\small {$\tilde{\tilde{\bm{x}}}$}}}
\htext (8.3 -1.7){{\small $\tilde{z}$}}
\htext (7.1 -1.2){{\small {$\tilde{\bm{r}}$}}}
\htext (7.4 -2){{\small {$\tilde{\bm{s}}$}}}
\htext (7 -.5){{\small $t$}}

\htext (9 -1.2) {{\Large $\overset{T^{246}_{b,d,f}}{\longrightarrow}$}}

\move (12 0) \lvec(12.7 -0.7)\lvec(12.7 -1.7)
\move (12 0) \lvec (11 0)
\move(12 -1.4)\lvec(12.7 -0.7)\lvec(11.7 -0.7)\lvec(11 -1.4)
\lvec(11.7 -0.7)\lvec(11 0)
\move(12 -1.4)\lvec(11 -1.4)\lvec(11 -2.4)
\move(11 -1.4)\lvec(10.3 -0.7)
\move(12 -1.4)\lvec(12 -2.4)
\move(10.3 -0.7)\lvec(11 0)
\move(12 -2.4)\lvec(12.7 -1.7)
\move(10.3 -1.7)\lvec(10.3 -0.7)
\move(12 -2.4)\lvec(11 -2.4)\lvec(10.3 -1.7)

\htext (11.7 -0.4){{\small {$\tilde{\tilde{\bm{r}}}$}}}
\htext (10.5 -1.7){{\small $\tilde{\tilde{x}}$}}
\htext (12.3 -1.7){{\small $\tilde{z}$}}
\htext (11.7 -1.2){{\small {$\tilde{\bm{t}}$}}}
\htext (11.4 -2){{\small $\tilde{s}$}}
\htext (10.8 -.8){{\small  $\tilde{\tilde{\bm{y}}}$}}

\htext (13 -1.2) {{\Large $\overset{T^{356}_{c,e,f}}{\longrightarrow}$}}

\move (16 0) \lvec(16.7 -0.7)\lvec(16.7 -1.7)\lvec(15.7 -1.7)\lvec(15 -2.4)
\move(15.7 -1.7)\lvec(15.7 -0.7)
\move (16 0) \lvec (15 0)
\move(16.7 -0.7)\lvec(15.7 -0.7)\lvec(15 -1.4)
\lvec(15.7 -0.7)\lvec(15 0)
\move(15 -1.4)\lvec(15 -2.4)
\move(15 -1.4)\lvec(14.3 -0.7)
\move(14.3 -0.7)\lvec(15 0)
\move(16 -2.4)\lvec(16.7 -1.7)
\move(14.3 -1.7)\lvec(14.3 -0.7)
\move(16 -2.4)\lvec(15 -2.4)\lvec(14.3 -1.7)

\htext (15.7 -0.4){{\small $\tilde{\tilde{r}}$}}
\htext (14.5 -1.7){{\small $\tilde{\tilde{x}}$}}
\htext (16.1 -1.3){{\small $\tilde{\tilde{\bm{s}}}$}}
\htext (15.25 -1.7){{\small $\tilde{\tilde{\bm{z}}}$}}
\htext (15.8 -2.2){{\small $\tilde{\tilde{\bm{t}}}$}}
\htext (14.8 -.8){{\small $\tilde{\tilde{y}}$}}

\move (0 0) \lvec (0 -1)\lvec(0.7 -1.7)\lvec(0.7 -0.7)\lvec(0 0)
\move (0 0) \lvec (-1 0)
\move (0 0)\lvec(-0.7 -0.7)\lvec(-0.7 -1.7)\lvec(0 -1)
\move (-0.7 -0.7)\lvec(-1.7 -0.7)\lvec(-1 0)
\move(-0.7 -1.7)\lvec(0 -2.4)\lvec(0.7 -1.7)
\move(-0.7 -1.7)\lvec(-1.7 -1.7)\lvec(-1.7 -0.7)
\move(0 -2.4)\lvec(-1 -2.4)\lvec(-1.7 -1.7)

\move (0 -3) \lvec (0 -4)\lvec(0.7 -4.7)\lvec(0.7 -3.7)\lvec(0 -3)
\move (0 -3) \lvec (-1 -3)
\move (0 -3)\lvec(-0.7 -3.7)\lvec(-0.7 -4.7)\lvec(0 -4)
\move (-0.7 -3.7)\lvec(-1.7 -3.7)\lvec(-1 -3)
\move(-0.7 -4.7)\lvec(0 -5.4)\lvec(0.7 -4.7)
\move(-0.7 -4.7)\lvec(-1.7 -4.7)\lvec(-1.7 -3.7)
\move(0 -5.4)\lvec(-1 -5.4)\lvec(-1.7 -4.7)

\htext (-0.5 -4){{\small $z$}}
\htext (0.3 -4){{\small $x$}}
\htext (-0.1 -4.8){{\small $y$}}
\htext (-.9 -5.2){{\small $r$}}
\htext (-1.3 -4.3){{\small $s$}}
\htext (-1 -3.5){{\small $t$}}

\htext (1 -4.2) {{\Large $\overset{T^{356}_{c,e,f}}{\longrightarrow}$}}

\move (4 -3) \lvec (4 -4)\lvec(4.7 -4.7)\lvec(4.7 -3.7)\lvec(4 -3)
\move (4 -3) \lvec (3 -3)
\move (3.3 -4.7)\lvec(4 -4)\lvec(3 -4)\lvec(3 -3)
\move (3 -4)\lvec(2.3 -4.7)
\move (2.3 -3.7)\lvec(3 -3)
\move(3.3 -4.7)\lvec(4 -5.4)\lvec(4.7 -4.7)
\move(3.3 -4.7)\lvec(2.3 -4.7)\lvec(2.3 -3.7)
\move(4 -5.4)\lvec(3 -5.4)\lvec(2.3 -4.7)

\htext (3.1 -4.5){{\small $\hat{\bm{t}}$}}
\htext (4.3 -4){{\small $x$}}
\htext (3.9 -4.8){{\small $y$}}
\htext (3.1 -5.2){{\small $r$}}
\htext (2.5 -4){{\small  $\hat{\bm{z}}$}}
\htext (3.4 -3.5){{\small $\hat{\bm{s}}$}}

\htext (5 -4.2) {{\Large $\overset{T^{246}_{b,d,f}}{\longrightarrow}$}}

\move (8 -3) \lvec (8 -4)\lvec(8.7 -4.7)\lvec(8.7 -3.7)\lvec(8 -3)
\move (8 -3) \lvec (7 -3)
\move (8 -4)\lvec(7 -4)\lvec(7 -3)
\move (7 -4)\lvec(6.3 -4.7)
\move (6.3 -3.7)\lvec(7 -3)
\move(8 -5.4)\lvec(8.7 -4.7)
\move(6.3 -4.7)\lvec(6.3 -3.7)
\move(8 -5.4)\lvec(7 -5.4)\lvec(6.3 -4.7)
\move (8.7 -4.7)\lvec(7.7 -4.7)\lvec (7 -5.4)
\move (7.7 -4.7)\lvec(7 -4)

\htext (6.9 -4.8){{\small $\hat{\bm{y}}$}}
\htext (8.3 -4){{\small $x$}}
\htext (7.7 -4.5){{\small $\hat{\bm{r}}$}} 
\htext (7.75 -5.2){{\small $\hat{\hat{\bm{t}}}$}} 
\htext (6.5 -4){{\small $\hat{z}$}}
\htext (7.4 -3.5){{\small $\hat{s}$}}

\htext (9 -4.2) {{\Large $\overset{T^{145}_{a,d,e}}{\longrightarrow}$}}

\move (12.7 -4.7)\lvec(12.7 -3.7)\lvec(12 -3)
\move(12.7 -3.7)\lvec(11.7 -3.7)\lvec(11.7 -4.7)
\move(11.7 -3.7)\lvec(11 -3)
\move (12 -3) \lvec (11 -3)
\move(11 -4)\lvec(11 -3)
\move (11 -4)\lvec(10.3 -4.7)
\move (10.3 -3.7)\lvec(11 -3)
\move(12 -5.4)\lvec(12.7 -4.7)
\move(10.3 -4.7)\lvec(10.3 -3.7)
\move(12 -5.4)\lvec(11 -5.4)\lvec(10.3 -4.7)
\move (12.7 -4.7)\lvec(11.7 -4.7)\lvec (11 -5.4)
\move (11.7 -4.7)\lvec(11 -4)

\htext (10.9 -4.8){{\small $\hat{y}$}}
\htext (12.1 -4.3){{\small $\hat{\hat{\bm{s}}}$}} 
\htext (11.2 -4){{\small $\hat{\bm{x}}$}} 
\htext (11.7 -5.2){{\small $\hat{\hat{t}}$}} 
\htext (10.5 -4){{\small $\hat{z}$}}
\htext (11.7 -3.5){{\small $\hat{\hat{\bm{r}}}$}}

\htext (13 -4.2) {{\Large $\overset{T^{123}_{a,b,c}}{\longrightarrow}$}}

\move (16 -3) \lvec(16.7 -3.7)\lvec(16.7 -4.7)\lvec(15.7 -4.7)\lvec(15 -5.4)
\move(15.7 -4.7)\lvec(15.7 -3.7)
\move (16 -3) \lvec (15 -3)
\move(16.7 -3.7)\lvec(15.7 -3.7)\lvec(15 -4.4)
\lvec(15.7 -3.7)\lvec(15 -3)
\move(15 -4.4)\lvec(15 -5.4)
\move(15 -4.4)\lvec(14.3 -3.7)
\move(14.3 -3.7)\lvec(15 -3)
\move(16 -5.4)\lvec(16.7 -4.7)
\move(14.3 -4.7)\lvec(14.3 -3.7)
\move(16 -5.4)\lvec(15 -5.4)\lvec(14.3 -4.7)

\htext (15.7 -3.4){{\small $\hat{\hat{r}}$}}
\htext (14.5 -4.7){{\small $\hat{\hat{\bm{x}}}$}}
\htext (16.1 -4.3){{\small $\hat{\hat{s}}$}}
\htext (15.25 -4.7){{\small $\hat{\hat{\bm{z}}}$}}
\htext (15.8 -5.2){{\small $\hat{\hat{t}}$}}
\htext (14.8 -3.8){{\small $\hat{\hat{\bm{y}}}$}}

\vtext (15.5 -2.9){$\pmb{=}$}
\vtext (-.5 -2.9){$\pmb{=}$}
}
\caption{Tetrahedron equation. Schematic representation (\cite{Doliwa-Kashaev}).}
\label{Tet_eq-rep}
\end{figure}

\begin{theorem}\label{six-factorisation}
Let $T_{a,b}:=(x,y,z)\rightarrow (u_{a,b,c}(x,y,z),v_{a,b,c}(x,y,z),w_{a,b,c}(x,y,z))$ be a map satisfying the local Yang--Baxter equation
$$
{\rm L}^3_{12}(u,a) {\rm L}^3_{13}(v,b) {\rm L}_{23}^3(w,c)= {\rm L}^3_{23}(z,c) {\rm L}^3_{13}(y,b) {\rm L}^3_{12}(x,a),
$$
for some matrix ${\rm L}={\rm L}(x,a)$. Then, if for this Lax matrix ${\rm L}$ the following matrix \textit{six-factorisation} problem
\begin{align}\label{6-fac}
&{\rm L}^4_{34}(\hat{t},a_6){\rm L}^4_{24}(\hat{s},a_5){\rm L}^4_{14}(\hat{r},a_4){\rm L}^4_{23}(\hat{z},a_3){\rm L}^4_{13}(\hat{y},a_2){\rm L}^4_{12}(\hat{x},a_1)=\nonumber\\
&{\rm L}^4_{34}(t,a_6){\rm L}^4_{24}(s,a_5){\rm L}^4_{14}(r,a_4){\rm L}^4_{23}(z,a_3){\rm L}^4_{13}(y,a_2){\rm L}^4_{12}(x,a_1)
\end{align}
implies the trivial solution $\hat{t}=t$, $\hat{s}=s$, $\hat{r}=r$, $\hat{z}=z$, $\hat{y}=y$ and $\hat{x}=x$, then map  $T_{a,b,c}:=(x,y,z)\rightarrow (u_{a,b,c}(x,y,z),v_{a,b,c}(x,y,z),w_{a,b,c}(x,y,z))$ is a parametric tetrahedron map.
\end{theorem}
\begin{proof}
According to the right-hand side of the tetrahedron equation, we have:
\begin{align}\label{6-fac-left}
   \left[{\rm L}^4_{12}(x,a_1){\rm L}^4_{13}(y,a_2) {\rm L}^4_{23}(z,a_3)\right]&{\rm L}^4_{14}(r,a_4) {\rm L}^4_{24}(s,a_5) {\rm L}^4_{34}(t,a_6)=\nonumber\\
  & ={\rm L}^4_{23}(\tilde{z},a_3){\rm L}^4_{13}(\tilde{y},a_2)\left[{\rm L}^4_{12}(\tilde{x},a_1) {\rm L}^4_{14}(r,a_4){\rm L}^4_{24}(s,a_5)\right] {\rm L}^4_{34}(t,a_6)\nonumber\\
  & ={\rm L}^4_{23}(\tilde{z},a_3){\rm L}^4_{13}(\tilde{y},a_2) {\rm L}^4_{24}(\tilde{s},a_5) {\rm L}^4_{14}(\tilde{r},a_4){\rm L}^4_{12}(\tilde{\tilde{x}},a_1)  {\rm L}^4_{34}(t,a_6)\nonumber\\
   & ={\rm L}^4_{23}(\tilde{z},a_3) {\rm L}^4_{24}(\tilde{s},a_5)\left[ {\rm L}^4_{13}(\tilde{y},a_2) {\rm L}^4_{14}(\tilde{r},a_4) {\rm L}^4_{34}(t,a_6)\right] {\rm L}^4_{12}(\tilde{\tilde{x}},a_1)\nonumber\\
   & =\left[{\rm L}^4_{23}(\tilde{z},a_3) {\rm L}^4_{24}(\tilde{s},a_5) {\rm L}^4_{34}(\tilde{t},a_6)\right] {\rm L}^4_{14}(\tilde{\tilde{r}},a_4){\rm L}^4_{13}(\tilde{\tilde{y}},a_2) {\rm L}^4_{12}(\tilde{\tilde{x}},a_1)\nonumber\\
   & = {\rm L}^4_{34}(\tilde{\tilde{t}},a_6) {\rm L}^4_{24}(\tilde{\tilde{s}},a_5) {\rm L}^4_{23}(\tilde{\tilde{z}},a_3) {\rm L}^4_{14}(\tilde{\tilde{r}},a_4){\rm L}^4_{13}(\tilde{\tilde{y}},a_2) {\rm L}^4_{12}(\tilde{\tilde{x}},a_1)\nonumber\\
   & = {\rm L}^4_{34}(\tilde{\tilde{t}},a_6) {\rm L}^4_{24}(\tilde{\tilde{s}},a_5) {\rm L}^4_{14}(\tilde{\tilde{r}},a_4) {\rm L}^4_{23}(\tilde{\tilde{z}},a_3) {\rm L}^4_{13}(\tilde{\tilde{y}},a_2) {\rm L}^4_{12}(\tilde{\tilde{x}},a_1),
\end{align}
where we have used Lemmas \ref{LijLji} and \ref{3-eq-equiv}. On the other hand, using the left-hand side of the tetrahedron equation, we obtain:

\begin{align}\label{6-fac-right}
   {\rm L}^4_{12}(x,a_1){\rm L}^4_{13}(y,a_2) {\rm L}^4_{23}(z,a_3)&{\rm L}^4_{14}(r,a_4) {\rm L}^4_{24}(s,a_5) {\rm L}^4_{34}(t,a_6)\nonumber\\
  &={\rm L}^4_{12}(x,a_1){\rm L}^4_{13}(y,a_2) {\rm L}^4_{14}(r,a_4) \left[{\rm L}^4_{23}(z,a_3) {\rm L}^4_{24}(s,a_5) {\rm L}^4_{34}(t,a_6)\right]\nonumber\\
  &={\rm L}^4_{12}(x,a_1)\left[{\rm L}^4_{13}(y,a_2) {\rm L}^4_{14}(r,a_4) {\rm L}^4_{34}(\hat{t},a_6)\right] {\rm L}^4_{24}(\hat{s},a_5) {\rm L}^4_{23}(\hat{z},a_3)\nonumber \\
  &={\rm L}^4_{12}(x,a_1) {\rm L}^4_{34}(\hat{\hat{t}},a_6) {\rm L}^4_{14}(\hat{r},a_4) {\rm L}^4_{13}(\hat{y},a_2) {\rm L}^4_{24}(\hat{s},a_5) {\rm L}^4_{23}(\hat{z},a_3) \nonumber\\
   &={\rm L}^4_{34}(\hat{\hat{t}},a_6) \left[{\rm L}^4_{12}(x,a_1) {\rm L}^4_{14}(\hat{r},a_4) {\rm L}^4_{24}(\hat{s},a_5)\right] {\rm L}^4_{13}(\hat{y},a_2) {\rm L}^4_{23}(\hat{z},a_3)\nonumber \\
   &={\rm L}^4_{34}(\hat{\hat{t}},a_6) {\rm L}^4_{24}(\hat{\hat{s}},a_5) {\rm L}^4_{14}(\hat{\hat{r}},a_4) \left[{\rm L}^4_{12}(\hat{x},a_1)  {\rm L}^4_{13}(\hat{y},a_2) {\rm L}^4_{23}(\hat{z},a_3)\right] \nonumber\\
    &={\rm L}^4_{34}(\hat{\hat{t}},a_6) {\rm L}^4_{24}(\hat{\hat{s}},a_5) {\rm L}^4_{14}(\hat{\hat{r}},a_4) {\rm L}^4_{23}(\hat{\hat{z}},a_3) {\rm L}^4_{13}(\hat{\hat{y}},a_2) {\rm L}^4_{12}(\hat{\hat{x}},a_1).
\end{align}

The left-hand sides of \eqref{6-fac-left} and \eqref{6-fac-right} are equal, therefore we obtain
\begin{align*}
    &{\rm L}^4_{34}(\tilde{\tilde{t}},a_6) {\rm L}^4_{24}(\tilde{\tilde{s}},a_5) {\rm L}^4_{14}(\tilde{\tilde{r}},a_4) {\rm L}^4_{23}(\tilde{\tilde{z}},a_3) {\rm L}^4_{13}(\tilde{\tilde{y}},a_2) {\rm L}^4_{12}(\tilde{\tilde{x}},a_1)=\\
    &{\rm L}^4_{34}(\hat{\hat{t}},a_6) {\rm L}^4_{24}(\hat{\hat{s}},a_5) {\rm L}^4_{14}(\hat{\hat{r}},a_4) {\rm L}^4_{23}(\hat{\hat{z}},a_3) {\rm L}^4_{13}(\hat{\hat{y}},a_2) {\rm L}^4_{12}(\hat{\hat{x}},a_1)
\end{align*}
Now, if the above implies that $\hat{\hat{x}}=\tilde{\tilde{x}}$, $\hat{\hat{y}}=\tilde{\tilde{y}}$, $\hat{\hat{z}}=\tilde{\tilde{z}}$, $\hat{\hat{r}}=\tilde{\tilde{r}}$, $\hat{\hat{s}}=\tilde{\tilde{s}}$ and $\hat{\hat{t}}=\tilde{\tilde{t}}$, then map $T_{a,b}$ satisfies the tetrahedron equation.
\end{proof}

For example, the noncommutative Hirota map \cite{Doliwa-Kashaev}
\begin{align*}
\bm{x}\mapsto\, &\bm{u}=\bm{y}\bm{x}(\bm{x}+\bm{z})^{-1},\\
\bm{y}\mapsto \,&\bm{v}=\bm{x}+\bm{z},\\
\bm{z}\mapsto \,&\bm{w}=\bm{y}\bm{z}(\bm{x}+\bm{z})^{-1},
\end{align*}
which appears in \cite{Doliwa-Kashaev} and satisfies the local Yang--Baxter \eqref{Lax-Tetra} equation for ${\rm L}(\bm{x})=\begin{pmatrix}\bm{x} & 1\\ 1 & 0 \end{pmatrix}$.

It was proven in \cite{Doliwa-Kashaev} that the above map satisfies the tetrahedron property. Now, we prove it using Theorem \ref{six-factorisation}. Substituting the latter matrix to equation \eqref{6-fac}, we obtain the following system of equations in noncommutative variables
\begin{equation}\label{nHm-eqs}
    \hat{\bm{r}}\hat{\bm{y}}\hat{\bm{x}}=\bm{r}\bm{y}\bm{x},\quad \hat{\bm{r}}\hat{\bm{y}}=\bm{r}\bm{y},\quad \hat{\bm{s}}(\hat{\bm{z}}+\hat{\bm{x}})+\hat{\bm{y}}\hat{\bm{x}}=\bm{s}(\bm{z}+\bm{x})+\bm{y}\bm{x}, \quad \hat{\bm{t}}+\hat{\bm{z}}+\hat{\bm{x}}=\bm{t}+\bm{z}+\bm{x},\quad \hat{\bm{s}}+\hat{\bm{y}}=\bm{s}+\bm{y},
\end{equation}
as well as $\hat{\bm{r}}=\bm{r}$, $\hat{\bm{y}}=\bm{y}$ and $\hat{\bm{s}}=\bm{s}$ in view of which equations \eqref{nHm-eqs} imply $\hat{\bm{x}}=\bm{x}$, $\hat{\bm{z}}=\bm{z}$ and $\hat{\bm{t}}=\bm{t}$.

\begin{remark}\normalfont  As it is demonstrated in the above example, with the help of Theorem \ref{six-factorisation} we can prove that a map satisfies the tetrahedron equation without using the map, but just using its Lax representation.
\end{remark}

\subsection{Matrix six-factorisation and correspondences}
In this section, we demonstrate that the six-factorisation problem can be also used for correspondences which satisfy the local Yang--Baxter equation. As illustrative examples, we construct some novel tetrahedron maps which are generated by matrices related to the lattice modified KdV equation and also to one-dimensional relativistic elastic collision of two particles.

Let us start with the gauge transformation $\psi\mapsto\bar{\psi}={\rm V}\psi$, ${\rm V}={\rm V}(v,U)=\begin{pmatrix} \frac{1}{v} & 0\\ \frac{U}{v} & 1 \end{pmatrix}$, related to the lattice mKdV equation \cite{Frank-Walker}. In fact, $v$ is a solution to the lattice mKdV equation, while $U$ satisfies the lattice KdV equation \cite{Frank-Walker}. We change $(v,U)\rightarrow (\bm{x}_1,\bm{x}_2)$, $\bm{x}_i\in\mathfrak{R}$, $i=1,2$, in the former matrix  ${\rm V}$, namely we consider the matrix ${\rm V}(\bm{x}_1,\bm{x}_2)=\begin{pmatrix} \bm{x}_1^{-1} & 0\\ \bm{x}_2\bm{x}_1^{-1} & 1 \end{pmatrix}$.

We now define the $3\times 3$ generalisations of matrix  ${\rm V}(\bm{x}_1,\bm{x}_2)$, namely the following:
\begin{align*}
    &{\rm V}^3_{12}=\begin{pmatrix}\bm{x}_1^{-1} & 0 & 0  \\ \bm{x}_2\bm{x}_1^{-1} & 1 & 0  \\ 0 & 0 & 1 \end{pmatrix}, \quad
    {\rm V}^3_{13}=\begin{pmatrix}\bm{x}_1^{-1} & 0 & 0  \\ 0 & 1 & 0  \\ \bm{x}_2\bm{x}_1^{-1} & 0 & 1  \end{pmatrix},\quad
    {\rm V}^3_{23}=\begin{pmatrix}1 & 0 & 0  \\ 0 & \bm{x}_1^{-1} & 0  \\ 0 & \bm{x}_2\bm{x}_1^{-1} & 1 \end{pmatrix}.
\end{align*}
Substitution to the local Yang--Baxter equation \eqref{Lax-Tetra} implies the following correspondence
\begin{equation}\label{corr-gauge-mKdV}
    \bm{u}_1=\bm{v}_1^{-1}\bm{x}_1\bm{y}_1,\quad \bm{u}_2=\bm{z}_1^{-1}\bm{x}_2\bm{y}_1
,\quad \bm{v}_2=(\bm{z}_2\bm{z}_1^{-1}\bm{x}_2+\bm{y}_2\bm{y}_1^{-1})\bm{x}_1^{-1}\bm{v}_1,\quad \bm{w}_1=\bm{z}_1,\quad \bm{w}_2=\bm{z}_2.\end{equation}
This correspondence satisfies the local Yang--Baxter equation for any $\bm{v}_1\in\mathfrak{R}$; nevertheless, it does not define a tetrahedron map for arbitrary choice of $\bm{v}_1\in\mathfrak{R}$. 

However, we have the following.

\begin{proposition}
For $\bm{v}_1=\bm{z}_1$, correspondence \eqref{corr-gauge-mKdV} defines the map
\begin{equation}\label{gauge-mKdV-map}
    (\bm{x}_1,\bm{x}_2,\bm{y}_1,\bm{y}_2,\bm{z}_1,\bm{z}_2)\overset{T_{a,b,c}}{\longrightarrow }(\bm{z}_1^{-1}\bm{x}_1\bm{y}_1,\bm{z}_1^{-1}\bm{x}_2\bm{y}_1,\bm{z}_1,(\bm{z}_2\bm{z}_1^{-1}\bm{x}_2+\bm{y}_2\bm{y}_1^{-1})\bm{x}_1^{-1}\bm{z}_1,\bm{z}_1,\bm{z}_2).
\end{equation}
Map \eqref{gauge-mKdV-map} is a noninvolutive, noncommutative, parametric tetrahedron map.
\end{proposition}
\begin{proof}
Map \eqref{gauge-mKdV-map} follows from \eqref{corr-gauge-mKdV} for $\bm{v}_1=\bm{z}_1\in\mathfrak{R}$.

Now, we define the $4\times 4$ generalisations of matrix  matrices ${\rm V}(\bm{x}_1,\bm{x}_2)$ as follows:
\begin{align*}
    &{\rm V}^4_{12}=\begin{pmatrix}\bm{x}_1^{-1} & 0 & 0 & 0 \\ \bm{x}_2\bm{x}_1^{-1} & 1 & 0 & 0 \\ 0 & 0 & 1 & 0 \\ 0 & 0 & 0 & 1\end{pmatrix}, \quad
    {\rm V}^4_{13}=\begin{pmatrix}\bm{x}_1^{-1} & 0 & 0 & 0 \\ 0 & 1 & 0 & 0 \\ \bm{x}_2\bm{x}_1^{-1} & 0 & 1 & 0 \\ 0 & 0 & 0 & 1\end{pmatrix},\quad
    {\rm V}^4_{23}=\begin{pmatrix}1 & 0 & 0 & 0 \\ 0 & \bm{x}_1^{-1} & 0 & 0 \\ 0 & \bm{x}_2\bm{x}_1^{-1} & 1 & 0 \\ 0 & 0 & 0 & 1\end{pmatrix},\\
    & {\rm V}^4_{14}=\begin{pmatrix}\bm{x}_1^{-1} & 0 & 0 & 0 \\ 0 & 1 & 0 & 0 \\ 0 & 0 & 1 & 0 \\ \bm{x}_2\bm{x}_1^{-1} & 0 & 0 & 1\end{pmatrix}, \quad
        {\rm V}^4_{24}=\begin{pmatrix}1 & 0 & 0 & 0 \\ 0 & \bm{x}_1^{-1} & 0 & 0 \\ 0 & 0 & 1 & 0 \\ 0 & \bm{x}_2\bm{x}_1^{-1} & 0 & 1\end{pmatrix}, \quad
    {\rm V}^4_{34}=\begin{pmatrix}1 & 0 & 0 & 0 \\ 0 & 1 & 0 & 0 \\ 0 & 0 & \bm{x}_1^{-1} & 0 \\ 0 & 0 & \bm{x}_2\bm{x}_1^{-1} & 1\end{pmatrix}.
\end{align*}

Regarding the tetrahedron property, using the left-hand side of the tetrahedron equation, and taking into account that $\bm{y}_1\mapsto\bm{z}_1, \bm{z}_1\mapsto\bm{z}_1$ and $\bm{z}_2\mapsto\bm{z}_2$ in \eqref{gauge-mKdV-map}, we have:
\begin{align}\label{6fac-gauge-mKdV-map-left}
   &{\rm V}^4_{12}(\bm{x}_1,\bm{x}_2,a_1){\rm V}^4_{13}(\bm{y}_1,\bm{y}_2,a_2) {\rm V}^4_{23}(\bm{z}_1,\bm{z}_2,a_3){\rm V}^4_{14}(\bm{r}_1,\bm{r}_2,a_4) {\rm V}^4_{24}(\bm{s}_1,\bm{s}_2,a_5) {\rm V}^4_{34}(\bm{t}_1,\bm{t}_2,a_6)=\nonumber\\
    &{\rm V}^4_{23}(\bm{z}_1,\bm{z}_2,a_3) {\rm V}^4_{13}(\bm{z}_1,\tilde{\bm{y}}_2,a_2) {\rm V}^4_{12}(\tilde{\bm{x}}_1,\tilde{\bm{x}}_2,a_1) {\rm V}^4_{14}(\bm{r}_1,\bm{r}_2,a_4) {\rm V}^4_{24}(\bm{s}_1,\bm{s}_2,a_5) {\rm V}^4_{34}(\bm{t}_1,\bm{t}_2,a_6)=\nonumber\\
    &{\rm V}^4_{23}(\bm{z}_1,\bm{z}_2,a_3) {\rm V}^4_{13}(\bm{z}_1,\tilde{\bm{y}}_2,a_2) {\rm V}^4_{24}(\bm{s}_1,\bm{s}_2,a_5) {\rm V}^4_{14}(\bm{s}_1,\tilde{\bm{r}}_2,a_4) {\rm V}^4_{12}(\tilde{\tilde{\bm{x}}}_1,\tilde{\tilde{\bm{x}}}_2,a_1)  {\rm V}^4_{34}(\bm{t}_1,\bm{t}_2,a_6)=\nonumber\\
    &{\rm V}^4_{23}(\bm{z}_1,\bm{z}_2,a_3) {\rm V}^4_{24}(\bm{s}_1,\bm{s}_2,a_5) {\rm V}^4_{13}(\bm{z}_1,\tilde{\bm{y}}_2,a_2) {\rm V}^4_{14}(\bm{s}_1,\tilde{\bm{r}}_2,a_4)  {\rm V}^4_{34}(\bm{t}_1,\bm{t}_2,a_6) {\rm V}^4_{12}(\tilde{\tilde{\bm{x}}}_1,\tilde{\tilde{\bm{x}}}_2,a_1)=\nonumber\\
     &{\rm V}^4_{23}(\bm{z}_1,\bm{z}_2,a_3) {\rm V}^4_{24}(\bm{s}_1,\bm{s}_2,a_5) {\rm V}^4_{34}(\bm{t}_1,\bm{t}_2,a_6) {\rm V}^4_{14}(\bm{t}_1,\tilde{\tilde{\bm{r}}}_2,a_4)  {\rm V}^4_{13}(\tilde{\bm{z}}_1,\tilde{\tilde{\bm{y}}}_2,a_2) {\rm V}^4_{12}(\tilde{\tilde{\bm{x}}}_1,\tilde{\tilde{\bm{x}}}_2,a_1)=\nonumber\\
      &{\rm V}^4_{34}(\bm{t}_1,\bm{t}_2,a_6) {\rm V}^4_{24}(\bm{t}_1,\tilde{\bm{s}}_2,a_5) {\rm V}^4_{23}(\tilde{\bm{z}}_1,\tilde{\bm{z}}_2,a_3) {\rm V}^4_{14}(\bm{t}_1,\tilde{\tilde{\bm{r}}}_2,a_4)  {\rm V}^4_{13}(\tilde{\bm{z}}_1,\tilde{\tilde{\bm{y}}}_2,a_2) {\rm V}^4_{12}(\tilde{\tilde{\bm{x}}}_1,\tilde{\tilde{\bm{x}}}_2,a_1)=\nonumber\\
       &{\rm V}^4_{34}(\bm{t}_1,\bm{t}_2,a_6) {\rm V}^4_{24}(\bm{t}_1,\tilde{\bm{s}}_2,a_5) {\rm V}^4_{14}(\bm{t}_1,\tilde{\tilde{\bm{r}}}_2,a_4)  {\rm V}^4_{23}(\tilde{\bm{z}}_1,\tilde{\bm{z}}_2,a_3) {\rm V}^4_{13}(\tilde{\bm{z}}_1,\tilde{\tilde{\bm{y}}}_2,a_2) {\rm V}^4_{12}(\tilde{\tilde{\bm{x}}}_1,\tilde{\tilde{\bm{x}}}_2,a_1).
\end{align}
On the other hand, using the left-hand side of the tetrahedron equation, and taking into account that $\bm{y}_1\mapsto\bm{z}_1, \bm{z}_1\mapsto\bm{z}_1$ and $\bm{z}_2\mapsto\bm{z}_2$ in \eqref{gauge-mKdV-map}, we obtain:
\begin{align}\label{6fac-gauge-mKdV-map-right}
   &{\rm V}^4_{12}(\bm{x}_1,\bm{x}_2,a_1){\rm V}^4_{13}(\bm{y}_1,\bm{y}_2,a_2) {\rm V}^4_{23}(\bm{z}_1,\bm{z}_2,a_3){\rm V}^4_{14}(\bm{r}_1,\bm{r}_2,a_4) {\rm V}^4_{24}(\bm{s}_1,\bm{s}_2,a_5) {\rm V}^4_{34}(\bm{t}_1,\bm{t}_2,a_6)=\nonumber\\
    &{\rm V}^4_{12}(\bm{x}_1,\bm{x}_2,a_1){\rm V}^4_{13}(\bm{y}_1,\bm{y}_2,a_2){\rm V}^4_{14}(\bm{r}_1,\bm{r}_2,a_4)  {\rm V}^4_{23}(\bm{z}_1,\bm{z}_2,a_3) {\rm V}^4_{24}(\bm{s}_1,\bm{s}_2,a_5) {\rm V}^4_{34}(\bm{t}_1,\bm{t}_2,a_6)=\nonumber\\
      &{\rm V}^4_{12}(\bm{x}_1,\bm{x}_2,a_1){\rm V}^4_{13}(\bm{y}_1,\bm{y}_2,a_2){\rm V}^4_{14}(\bm{r}_1,\bm{r}_2,a_4)  {\rm V}^4_{34}(\bm{t}_1,\bm{t}_2,a_6) {\rm V}^4_{24}(\bm{t}_1,\hat{\bm{s}}_2,a_5)  {\rm V}^4_{23}(\hat{\bm{z}}_1,\hat{\bm{z}}_2,a_3)=\nonumber\\
      &{\rm V}^4_{12}(\bm{x}_1,\bm{x}_2,a_1){\rm V}^4_{34}(\bm{t}_1,\bm{t}_2,a_6){\rm V}^4_{14}(\bm{t}_1,\hat{\bm{r}}_2,a_4){\rm V}^4_{13}(\hat{\bm{y}}_1,\hat{\bm{y}}_2,a_2) {\rm V}^4_{24}(\bm{t}_1,\hat{\bm{s}}_2,a_5)  {\rm V}^4_{23}(\hat{\bm{z}}_1,\hat{\bm{z}}_2,a_3)=\nonumber\\
      &{\rm V}^4_{34}(\bm{t}_1,\bm{t}_2,a_6){\rm V}^4_{12}(\bm{x}_1,\bm{x}_2,a_1){\rm V}^4_{14}(\bm{t}_1,\hat{\bm{r}}_2,a_4){\rm V}^4_{24}(\bm{t}_1,\hat{\bm{s}}_2,a_5){\rm V}^4_{13}(\hat{\bm{y}}_1,\hat{\bm{y}}_2,a_2)   {\rm V}^4_{23}(\hat{\bm{z}}_1,\hat{\bm{z}}_2,a_3)=\nonumber\\
       &{\rm V}^4_{34}(\bm{t}_1,\bm{t}_2,a_6){\rm V}^4_{24}(\bm{t}_1,\hat{\bm{s}}_2,a_5){\rm V}^4_{14}(\bm{t}_1,\hat{\hat{\bm{r}}}_2,a_4){\rm V}^4_{12}(\hat{\bm{x}}_1,\hat{\bm{x}}_2,a_1){\rm V}^4_{13}(\hat{\bm{y}}_1,\hat{\bm{y}}_2,a_2)   {\rm V}^4_{23}(\hat{\bm{z}}_1,\hat{\bm{z}}_2,a_3)=\nonumber\\
        &{\rm V}^4_{34}(\bm{t}_1,\bm{t}_2,a_6){\rm V}^4_{24}(\bm{t}_1,\hat{\bm{s}}_2,a_5){\rm V}^4_{14}(\bm{t}_1,\hat{\hat{\bm{r}}}_2,a_4){\rm V}^4_{23}(\hat{\bm{z}}_1,\hat{\bm{z}}_2,a_3){\rm V}^4_{13}(\hat{\bm{z}}_1,\hat{\hat{\bm{y}}}_2,a_2){\rm V}^4_{12}(\hat{\hat{\bm{x}}}_1,\hat{\hat{\bm{x}}}_2,a_1).
\end{align}

From the relations \eqref{6fac-gauge-mKdV-map-left} and \eqref{6fac-gauge-mKdV-map-right} it follows that
\begin{align*}
    &{\rm V}^4_{34}(\bm{t}_1,\bm{t}_2,a_6) {\rm V}^4_{24}(\bm{t}_1,\tilde{\bm{s}}_2,a_5) {\rm V}^4_{14}(\bm{t}_1,\tilde{\tilde{\bm{r}}}_2,a_4)  {\rm V}^4_{23}(\tilde{\bm{z}}_1,\tilde{\bm{z}}_2,a_3) {\rm V}^4_{13}(\tilde{\bm{z}}_1,\tilde{\tilde{\bm{y}}}_2,a_2) {\rm V}^4_{12}(\tilde{\tilde{\bm{x}}}_1,\tilde{\tilde{\bm{x}}}_2,a_1)=\\
    &{\rm V}^4_{34}(\bm{t}_1,\bm{t}_2,a_6){\rm V}^4_{24}(\bm{t}_1,\hat{\bm{s}}_2,a_5){\rm V}^4_{14}(\bm{t}_1,\hat{\hat{\bm{r}}}_2,a_4){\rm V}^4_{23}(\hat{\bm{z}}_1,\hat{\bm{z}}_2,a_3){\rm V}^4_{13}(\hat{\bm{z}}_1,\hat{\hat{\bm{y}}}_2,a_2){\rm V}^4_{12}(\hat{\hat{\bm{x}}}_1,\hat{\hat{\bm{x}}}_2,a_1),
\end{align*}
which implies the following system
\begin{subequations}\label{six-fac-Vij}
\begin{align}
    &\tilde{\bm{z}}_1=\hat{\bm{z}}_1,\quad \bm{t}_1\tilde{\bm{z}}_1\tilde{\tilde{\bm{x}}}_1=\bm{t}_1\hat{\bm{z}}_1\hat{\hat{\bm{x}}}_1,\quad \bm{t}_1^{-1}\tilde{\bm{z}}_1\tilde{\tilde{x}}_2\tilde{\tilde{x}}_1^{-1}=\bm{t}_1^{-1}\hat{\bm{z}}_1\hat{\hat{x}}_2\hat{\hat{x}}_1^{-1}, \quad \bm{t}_1^{-1}\tilde{\bm{z}}_2\bm{\tilde{z}}_1^{-1}=\bm{t}_1^{-1}\hat{\bm{z}}_2\bm{\hat{z}}_1^{-1},\label{six-fac-Vij-a}\\
    &\bm{t}_1^{-1}(\tilde{\bm{z}}_2\tilde{\bm{z}}_1^{-1}\tilde{\tilde{\bm{x}}}_2\tilde{\tilde{x}}_1^{-1}+\tilde{\tilde{y}}_2\tilde{z}_1^{-1}\tilde{\tilde{x}}_1)=\bm{t}_1^{-1}(\hat{\bm{z}}_2\hat{\bm{z}}_1^{-1}\hat{\hat{\bm{x}}}_2\hat{\hat{x}}_1^{-1}+\hat{\hat{y}}_2\hat{z}_1^{-1}\hat{\hat{x}}_1),\label{six-fac-Vij-b}\\
   & \bm{t}_2\bm{t}_1^{-1}\tilde{\bm{z}}_2\tilde{\bm{z}}_1^{-1}+\tilde{\bm{s}}_2\bm{t}_1^{-1}\tilde{\bm{z}}_1=\bm{t}_2\bm{t}_1^{-1}\hat{\bm{z}}_2\hat{\bm{z}}_1^{-1}+\hat{\bm{s}}_2\bm{t}_1^{-1}\hat{\bm{z}}_1,\label{six-fac-Vij-c}\\
   &\bm{t}_2\bm{t}_1^{-1}(\tilde{\bm{z}}_2\tilde{\bm{z}}_1^{-1}\tilde{\tilde{\bm{x}}}_2\tilde{\tilde{\bm{x}}}_1^{-1}+\tilde{\tilde{\bm{y}}}_2\tilde{\bm{z}}_1^{-1}\tilde{\tilde{\bm{x}}}_1)+\tilde{\bm{s}}_2\bm{t}_1^{-1}\tilde{\bm{z}}_1\tilde{\tilde{\bm{x}}}_2\tilde{\tilde{\bm{x}}}_1^{-1}+\tilde{\tilde{\bm{r}}}_2\bm{t}_1^{-1}\tilde{\bm{z}}_1\tilde{\tilde{x}}_1=\nonumber\\
    &\bm{t}_2\bm{t}_1^{-1}(\hat{\bm{z}}_2\hat{\bm{z}}_1^{-1}\hat{\hat{\bm{x}}}_2\hat{\hat{\bm{x}}}_1^{-1}+\hat{\hat{\bm{y}}}_2\hat{\bm{z}}_1^{-1}\hat{\hat{\bm{x}}}_1)+\hat{\bm{s}}_2\bm{t}_1^{-1}\hat{\bm{z}}_1\hat{\hat{\bm{x}}}_2\hat{\hat{\bm{x}}}_1^{-1}+\hat{\hat{\bm{r}}}_2\bm{t}_1^{-1}\hat{\bm{z}}_1\hat{\hat{x}}_1.\label{six-fac-Vij-d}
\end{align}
\end{subequations}
Since $\tilde{\bm{z}}_1=\hat{\bm{z}}_1$, from the rest of \eqref{six-fac-Vij-a} follows that $\tilde{\tilde{\bm{x}}}_1=\hat{\hat{\bm{x}}}_1$, $\tilde{\tilde{\bm{x}}}_2=\hat{\hat{\bm{x}}}_2$ and $\tilde{\bm{z}}_2=\hat{\bm{z}}_2$. Given that, \eqref{six-fac-Vij-b} and \eqref{six-fac-Vij-c} imply $\tilde{\tilde{\bm{y}}}_2=\hat{\hat{\bm{y}}}_2$ and $\tilde{\bm{s}}_2=\hat{\bm{s}}_2$, respectively, in view of which from \eqref{six-fac-Vij-d} we obtain that $\tilde{\tilde{\bm{r}}}_2=\hat{\hat{\bm{r}}}_2$. Therefore, map \eqref{gauge-mKdV-map} is a parametric tetrahedron map.

The noninvolutivity of map \eqref{gauge-mKdV-map} follows, for instance, from $\bm{u}_1\circ T_{a,b,c}=\bm{x}_1\bm{y}_1\bm{z}_1^{-1}\neq\bm{x}_1$.
\end{proof}

Now, we will derive another novel, noncommutative, parametric tetrahedron map generated by the matrix ${\rm L}_{\lambda}(\bm{x}_{1},\bm{x}_{2})=\begin{pmatrix}\lambda & k\bm{x}_1\\ k\bm{x}_2^{-1} & \lambda \end{pmatrix}$, where $\bm{x}_{i}\in\mathfrak{R}$, $i=1,2$, and $k\in Z(\mathfrak{R})$. For $\bm{x}_1=\bm{x}_2=x$, this matrix is a Lax matrix for a Yang--Baxter map which was derived from a system of one-dimensional relativistic elastic collision of two particles in \cite{Kouloukas}.

We define the matrices
\begin{align*}
    &{\rm L}^3_{0,12}=\begin{pmatrix}0 & k\bm{x}_1 & 0 \\ k\bm{x}_2^{-1} & 0 & 0  \\ 0 & 0 & 1 \end{pmatrix}, \quad
    {\rm L}^3_{0,13}=\begin{pmatrix}0 & 0 & k\bm{x}_1 \\ 0 & 1 & 0 \\ k\bm{x}_2^{-1} & 0 & 0 \end{pmatrix},\quad
    {\rm L}^3_{0,23}=\begin{pmatrix}1 & 0 & 0 \\ 0 & 0 & k\bm{x}_1  \\ 0 & k\bm{x}_2^{-1} & 0 \end{pmatrix}.
\end{align*}
Substitution to the local Yang--Baxter equation \eqref{Lax-Tetra} implies the following correspondence
\begin{equation}\label{relativ-corr}
    \bm{v}_1=\bm{u}_2\bm{z}_1\bm{y}_2^{-1}\bm{x}_1\bm{w}_2,\quad \bm{v}_2=a^{-1}bc^{-1}\bm{x}_2\bm{z}_2,\quad \bm{w}_1 =a^{-1}bc^{-1}\bm{u}_1^{-1}\bm{y}_1.
\end{equation}
The above correspondence satisfies the local Yang--Baxter equation for arbitrary $\bm{u}_1, \bm{u}_2, \bm{w}_2\in\mathfrak{R}$. But, not for any choice of the latter the associated map satisfies the tetrahedron equation.

We have the following.

\begin{proposition}
Correspondence \eqref{relativ-corr} defines the map
\begin{equation}\label{1d-rel-map} 
    (\bm{x}_1,\bm{x}_2,\bm{y}_1,\bm{y}_2,\bm{z}_1,\bm{z}_2)\overset{T_{a,b,c}}{\longrightarrow }(\bm{x}_1,\bm{x}_2,\bm{x}_2\bm{z}_1\bm{y}_2^{-1}\bm{x}_1\bm{z}_2,a^{-1}bc^{-1}\bm{x}_2\bm{z}_2,a^{-1}bc^{-1}\bm{x}_1^{-1}\bm{y}_1,\bm{z}_2).
\end{equation}
Map \eqref{1d-rel-map} is a noninvolutive, noncommutative, parametric tetrahedron map.
\end{proposition}

\begin{proof}
Obviously, map \eqref{1d-rel-map} is derived from correspondence \eqref{relativ-corr} for the choice $\bm{u}_1=\bm{x}_1, \bm{u}_2=\bm{x}_2$ and $\bm{w}_2=\bm{z}_2$.

Regarding the tetrahedron property, we first define the $4\times 4$ generalisations of matrix ${\rm L}_{\lambda}(\bm{x}_{1},\bm{x}_{2})$: 
\begin{align*}
    &{\rm L}^4_{0,12}=\begin{pmatrix}0 & k\bm{x}_1 & 0 & 0 \\ k\bm{x}_2^{-1} & 0 & 0 & 0 \\ 0 & 0 & 1 & 0 \\ 0 & 0 & 0 & 1\end{pmatrix}, \quad
    {\rm L}^4_{0,13}=\begin{pmatrix}0 & 0 & k\bm{x}_1 & 0 \\ 0 & 1 & 0 & 0 \\ k\bm{x}_2^{-1} & 0 & 0 & 0 \\ 0 & 0 & 0 & 1\end{pmatrix},\quad
    {\rm L}^4_{0,23}=\begin{pmatrix}1 & 0 & 0 & 0 \\ 0 & 0 & k\bm{x}_1 & 0 \\ 0 & k\bm{x}_2^{-1} & 0 & 0 \\ 0 & 0 & 0 & 1\end{pmatrix},\\
    \\
   & {\rm L}^4_{0,14}=\begin{pmatrix}0 & 0 & 0 & k\bm{x}_1 \\ 0 & 1 & 0 & 0 \\ 0 & 0 & 1 & 0 \\ k\bm{x}_2^{-1} & 0 & 0 & 0\end{pmatrix}, \quad
        {\rm L}^4_{0,24}=\begin{pmatrix}1 & 0 & 0 & 0 \\ 0 & 0 & 0 & k\bm{x}_1 \\ 0 & 0 & 1 & 0 \\ 0 & k\bm{x}_2^{-1} & 0 & 0\end{pmatrix}, \quad
    {\rm L}^4_{0,34}=\begin{pmatrix}1 & 0 & 0 & 0 \\ 0 & 1 & 0 & 0 \\ 0 & 0 & 0 & k\bm{x}_1 \\ 0 & 0 & k\bm{x}_2^{-1} & 0\end{pmatrix}.
\end{align*}

Now, using the left-hand side of the tetrahedron equation, and taking into account \eqref{1d-rel-map}, we have:{\small
\begin{align}\label{6fac-1d-rel-left}
   &{\rm L}^4_{0,12}(\bm{x}_1,\bm{x}_2,a_1){\rm L}^4_{0,13}(\bm{y}_1,\bm{y}_2,a_2) {\rm L}^4_{0,23}(\bm{z}_1,\bm{z}_2,a_3){\rm L}^4_{0,14}(\bm{r}_1,\bm{r}_2,a_4) {\rm L}^4_{0,24}(\bm{s}_1,\bm{s}_2,a_5) {\rm L}^4_{0,34}(\bm{t}_1,\bm{t}_2,a_6)=\nonumber\\
   &{\rm L}^4_{0,23}(a_1^{-1}a_2a_3^{-1}\bm{x}_1^{-1}\bm{y}_1,\bm{z}_2,a_3) {\rm L}^4_{0,13}(\bm{x}_2\bm{z}_1\bm{y}_2^{-1},a_1^{-1}a_2a_3^{-1}\bm{x}_2\bm{z}_2,a_2) {\rm L}^4_{0,12}(\bm{x}_1,\bm{x}_2,a_1) \cdot\nonumber\\
   &{\rm L}^4_{0,14}(\bm{r}_1,\bm{r}_2,a_4) {\rm L}^4_{0,24}(\bm{s}_1,\bm{s}_2,a_5) {\rm L}^4_{0,34}(\bm{t}_1,\bm{t}_2,a_6)=\nonumber\\
   &{\rm L}^4_{0,23}(a_1^{-1}a_2a_3^{-1}\bm{x}_1^{-1}\bm{y}_1,\bm{z}_2,a_3) {\rm L}^4_{0,13}(\bm{x}_2\bm{z}_1\bm{y}_2^{-1},a_1^{-1}a_2a_3^{-1}\bm{x}_2\bm{z}_2,a_2) {\rm L}^4_{0,24}(\tilde{\bm{s}}_1,\bm{s}_2,a_5)\cdot\nonumber\\
   &{\rm L}^4_{0,14}(\tilde{\bm{r}}_1,a_1^{-1}a_4a_5^{-1}\bm{x}_2\bm{s}_2,a_4) {\rm L}^4_{0,12}(\bm{x}_1,\bm{x}_2,a_1) {\rm L}^4_{0,34}(\bm{t}_1,\bm{t}_2,a_6)=\nonumber\\
    &{\rm L}^4_{0,23}(a_1^{-1}a_2a_3^{-1}\bm{x}_1^{-1}\bm{y}_1,\bm{z}_2,a_3) {\rm L}^4_{0,24}(\tilde{\bm{s}}_1,\bm{s}_2,a_5) {\rm L}^4_{0,13}(\bm{x}_2\bm{z}_1\bm{y}_2^{-1},a_1^{-1}a_2a_3^{-1}\bm{x}_2\bm{z}_2,a_2)\cdot\nonumber\\
   &{\rm L}^4_{0,14}(\tilde{\bm{r}}_1,a_1^{-1}a_4a_5^{-1}\bm{x}_2\bm{s}_2,a_4) {\rm L}^4_{0,34}(\bm{t}_1,\bm{t}_2,a_6){\rm L}^4_{0,12}(\bm{x}_1,\bm{x}_2,a_1) =\nonumber\\
   &{\rm L}^4_{0,23}(a_1^{-1}a_2a_3^{-1}\bm{x}_1^{-1}\bm{y}_1,\bm{z}_2,a_3) {\rm L}^4_{0,24}(\tilde{\bm{s}}_1,\bm{s}_2,a_5) {\rm L}^4_{0,34}(\tilde{\bm{t}}_1,\bm{t}_2,a_6)\cdot\nonumber\\
   &{\rm L}^4_{0,14}(\tilde{\tilde{\bm{r}}}_1,a_1^{-1}a_3^{-1}a_4a_6^{-1}\bm{x}_2\bm{z}_2\bm{t}_2,a_4) {\rm L}^4_{0,13}(\bm{x}_2\bm{z}_1\bm{y}_2^{-1}\bm{x}_1\bm{z}_2,a_1^{-1}a_2a_3^{-1}\bm{x}_2\bm{z}_2,a_2){\rm L}^4_{0,12}(\bm{x}_1,\bm{x}_2,a_1) =\nonumber\\
   &{\rm L}^4_{0,34}(\tilde{\tilde{\bm{t}}}_1,\bm{t}_2,a_6) {\rm L}^4_{0,24}(\tilde{\tilde{\bm{s}}}_1,a_3^{-1}a_5a_6^{-1}\bm{z}_2\bm{t}_2,a_5) {\rm L}^4_{0,23}(a_1^{-1}a_2a_3^{-1}\bm{x}_1^{-1}\bm{y}_1,\bm{z}_2,a_3)\cdot\nonumber\\
   &{\rm L}^4_{0,14}(\tilde{\tilde{\bm{r}}}_1,a_1^{-1}a_3^{-1}a_4a_6^{-1}\bm{x}_2\bm{z}_2\bm{t}_2,a_4) {\rm L}^4_{0,13}(\bm{x}_2\bm{z}_1\bm{y}_2^{-1}\bm{x}_1\bm{z}_2,a_1^{-1}a_2a_3^{-1}\bm{x}_2\bm{z}_2,a_2){\rm L}^4_{0,12}(\bm{x}_1,\bm{x}_2,a_1) =\nonumber\\
    &{\rm L}^4_{0,34}(\tilde{\tilde{\bm{t}}}_1,\bm{t}_2,a_6) {\rm L}^4_{0,24}(\tilde{\tilde{\bm{s}}}_1,a_3^{-1}a_5a_6^{-1}\bm{z}_2\bm{t}_2,a_5){\rm L}^4_{0,14}(\tilde{\tilde{\bm{r}}}_1,a_1^{-1}a_3^{-1}a_4a_6^{-1}\bm{x}_2\bm{z}_2\bm{t}_2,a_4) \cdot\nonumber\\
   &{\rm L}^4_{0,23}(a_1^{-1}a_2a_3^{-1}\bm{x}_1^{-1}\bm{y}_1,\bm{z}_2,a_3) {\rm L}^4_{0,13}(\bm{x}_2\bm{z}_1\bm{y}_2^{-1}\bm{x}_1\bm{z}_2,a_1^{-1}a_2a_3^{-1}\bm{x}_2\bm{z}_2,a_2){\rm L}^4_{0,12}(\bm{x}_1,\bm{x}_2,a_1).
\end{align}
}
On the other hand, using the right-hand side of the tetrahedron equation, and taking into account \eqref{1d-rel-map}, after a few steps we obtain:{\small
\begin{align}\label{6fac-1d-rel-right}
   &{\rm L}^4_{0,12}(\bm{x}_1,\bm{x}_2,a_1){\rm L}^4_{0,13}(\bm{y}_1,\bm{y}_2,a_2) {\rm L}^4_{0,23}(\bm{z}_1,\bm{z}_2,a_3){\rm L}^4_{0,14}(\bm{r}_1,\bm{r}_2,a_4) {\rm L}^4_{0,24}(\bm{s}_1,\bm{s}_2,a_5) {\rm L}^4_{0,34}(\bm{t}_1,\bm{t}_2,a_6)=\nonumber\\
    &{\rm L}^4_{0,34}(\hat{\hat{\bm{t}}}_1,\bm{t}_2,a_6) {\rm L}^4_{0,24}(\hat{\hat{\bm{s}}}_1,a_3^{-1}a_5a_6^{-1}\bm{z}_2\bm{t}_2,a_5){\rm L}^4_{0,14}(\hat{\hat{\bm{r}}}_1,a_1^{-1}a_3^{-1}a_4a_6^{-1}\bm{x}_2\bm{z}_2\bm{t}_2,a_4) \cdot\nonumber\\
   &{\rm L}^4_{0,23}(a_1^{-1}a_2a_3^{-1}\bm{x}_1^{-1}\bm{y}_1,\bm{z}_2,a_3) {\rm L}^4_{0,13}(\bm{x}_2\bm{z}_1\bm{y}_2^{-1}\bm{x}_1\bm{z}_2,a_1^{-1}a_2a_3^{-1}\bm{x}_2\bm{z}_2,a_2){\rm L}^4_{0,12}(\bm{x}_1,\bm{x}_2,a_1).
\end{align}
}

From equations \eqref{6fac-1d-rel-left} and \eqref{6fac-1d-rel-right} follows that
\begin{align}\label{6fac-1d-rel-sys}
   &{\rm L}^4_{0,34}(\tilde{\tilde{\bm{t}}}_1,\bm{t}_2,a_6) {\rm L}^4_{0,24}(\tilde{\tilde{\bm{s}}}_1,a_3^{-1}a_5a_6^{-1}\bm{z}_2\bm{t}_2,a_5){\rm L}^4_{0,14}(\tilde{\tilde{\bm{r}}}_1,a_1^{-1}a_3^{-1}a_4a_6^{-1}\bm{x}_2\bm{z}_2\bm{t}_2,a_4) \cdot\nonumber\\
   &{\rm L}^4_{0,23}(a_1^{-1}a_2a_3^{-1}\bm{x}_1^{-1}\bm{y}_1,\bm{z}_2,a_3) {\rm L}^4_{0,13}(\bm{x}_2\bm{z}_1\bm{y}_2^{-1}\bm{x}_1\bm{z}_2,a_1^{-1}a_2a_3^{-1}\bm{x}_2\bm{z}_2,a_2){\rm L}^4_{0,12}(\bm{x}_1,\bm{x}_2,a_1)=\nonumber\\
    &{\rm L}^4_{0,34}(\hat{\hat{\bm{t}}}_1,\bm{t}_2,a_6) {\rm L}^4_{0,24}(\hat{\hat{\bm{s}}}_1,a_3^{-1}a_5a_6^{-1}\bm{z}_2\bm{t}_2,a_5){\rm L}^4_{0,14}(\hat{\hat{\bm{r}}}_1,a_1^{-1}a_3^{-1}a_4a_6^{-1}\bm{x}_2\bm{z}_2\bm{t}_2,a_4) \cdot\nonumber\\
   &{\rm L}^4_{0,23}(a_1^{-1}a_2a_3^{-1}\bm{x}_1^{-1}\bm{y}_1,\bm{z}_2,a_3) {\rm L}^4_{0,13}(\bm{x}_2\bm{z}_1\bm{y}_2^{-1}\bm{x}_1\bm{z}_2,a_1^{-1}a_2a_3^{-1}\bm{x}_2\bm{z}_2,a_2){\rm L}^4_{0,12}(\bm{x}_1,\bm{x}_2,a_1),
\end{align}
which implies
$$
(\tilde{\tilde{\bm{s}}}_1-\hat{\hat{\bm{s}}}_1)(\bm{x}_2\bm{z}_2\bm{t}_2)^{-1}\bm{x}_2\bm{z}_1\bm{y}_2^{-1}\bm{x}_1\bm{z}_2=0,\quad (\tilde{\tilde{\bm{t}}}_1-\hat{\hat{\bm{t}}}_1)(\bm{z}_2\bm{t}_2)^{-1}\bm{x}_1^{-1}\bm{y}_1(\bm{x}_2\bm{z}_2)^{-1}\bm{x}_1,\quad a_4(\tilde{\tilde{\bm{r}}}_1-\hat{\hat{\bm{r}}}_1)=0,
$$
from where it follows that $\tilde{\tilde{\bm{s}}}_1=\hat{\hat{\bm{s}}}_1$, $\tilde{\tilde{\bm{t}}}_1=\hat{\hat{\bm{t}}}_1$ and $\tilde{\tilde{\bm{r}}}_1=\hat{\hat{\bm{r}}}_1$. That is, map \eqref{1d-rel-map} is a parametric tetrahedron map.

Map \eqref{1d-rel-map} is noninvolutive, since, for instance, $\bm{v}_1\circ T_{a,b,c}=\bm{x}_2\bm{x}_1^{-1}\bm{y}_1\bm{z}_2^{-1}\bm{x}_2^{-1}\bm{x}_1\bm{z}_2\neq\bm{y}_1$.
\end{proof}

\begin{remark}\normalfont
In this section, we saw how can tetrahedron maps be derived from correspondences which satisfy the local Yang--Baxter equation by fixing their free variables. The natural question arises as to whether the choice of the variables we made is unique. The answer is negative: as it was demonstrated in \cite{Sokor-2020}, different tetrahedron maps can be derived for various choices of the free variables. However, a correspondence satisfying the local Yang--Baxter equation does not define a tetrahedron map for arbitrary choice of the free variables. To check whether a choice of the free variable is `correct' one needs to verify whether the matrix six-factorisation problem \eqref{6-fac} implies the trivial solution. There is no systematic way of finding the particular values that define tetrahedron maps. 
\end{remark}

\section{Noncommutative NLS type tetrahedron maps}\label{NLS-noncomm_tetrahedron_maps}
A method for constructing parametric tetrahedron maps using Darboux transformations was presented in \cite{Sokor-2020}. In this section, we employ a noncommutative Darboux matrix derived in \cite{Sokor-Xenitidis} for the NLS equation in order to derive a noncommutative parametric tetrahedron map via the local Yang--Baxter equation. 

\subsection{A noncommutative NLS type Darboux transformation}
Consider the spatial part of the Lax pair of the noncommutative NLS system \cite{Sokolov}, namely the operator:
\begin{equation}\label{Lax-NLS-NC}
    \mathcal{L}=D_x-{\rm U}(\bm{p},\bm{q};\lambda),\quad {\rm{U}}=\lambda\begin{pmatrix}
   1 & 0\\
   0 & -1
\end{pmatrix}
+
\begin{pmatrix}
   0 & 2\bm{p}\\
   2\bm{q} & 0
\end{pmatrix}.
\end{equation}

The definition of a Darboux transformation for such operators is a transformation which leaves covariant the operator, namely
$$
\rm{M}\big(\mathcal{D}_x+\rm{U}(\rm{U};\lambda)\big)\rm{M}^{-1}=
\mathcal{D}_x-\rm{U}({\rm{U}_{10}};\lambda),
$$
where $\rm{M}$ is an invertible matrix called a \textit{Darboux matrix}. A Darboux matrix related to operator \eqref{Lax-NLS-NC} is the following \cite{Sokor-Xenitidis}:
\begin{equation}\label{DT-10-deg}
{\rm{M}}(\bm{f},\bm{p};a)=\lambda\begin{pmatrix}
1 & 0 \\
0 & 0
\end{pmatrix}
+
\begin{pmatrix}
\bm{f} & \bm{p} \\
a\bm{p}^{-1} & 0
\end{pmatrix}, 
\quad \bm{f}=\frac{1}{2}\bm{p}_x\bm{p}^{-1}.
\end{equation}

\subsection{A noncommutative NLS type tetahedron map}
We change $(\bm{f}+\lambda,\bm{p})\rightarrow (\bm{x}_1,\bm{x}_2)$, $\bm{x}_i\in\mathfrak{R}$, $i=1,2$, in \eqref{DT-10-deg}, namely we consider the matrix
\begin{equation}\label{DT-NLS}
{\rm{M}}(\bm{x}_1,\bm{x}_2;a)=
\begin{pmatrix}
\bm{x}_1 & \bm{x}_2 \\
a\bm{x}^{-1}_2 & 0
\end{pmatrix},
\end{equation}
where $a\in Z(\mathfrak{R})$.

We employ this matrix to construct a noncommutative tetrahedron map. We define the following matrices
{\small $$
  M^3_{12}(\bm{x}_1,\bm{x}_2;a)=\begin{pmatrix} 
 \bm{x}_1 &  \bm{x}_2 & 0\\ 
a\bm{x}_2^{-1} &  0 & 0\\
0 & 0 & 1
\end{pmatrix},\quad
 M^3_{13}(\bm{x}_1,\bm{x}_2;a)= \begin{pmatrix} 
 \bm{x}_1 & 0 & \bm{x}_2\\ 
0 & 1 & 0\\
a\bm{x}_2^{-1} & 0 & 0
\end{pmatrix}, \quad
 M^3_{23}(\bm{x}_1,\bm{x}_2;a)=\begin{pmatrix} 
   1 & 0 & 0 \\
0 & \bm{x}_1 & \bm{x}_2\\ 
0 & a\bm{x}_2^{-1} & 0
\end{pmatrix},
$$}
and substitute them to the local Yang--Baxter equation
$$
    M^3_{12}(\bm{u}_1,\bm{u}_2;a)M^3_{13}(\bm{v}_1,\bm{v}_2;b)M^3_{23}(\bm{w}_1,\bm{w}_2;c)= M^3_{23}(\bm{z}_1,\bm{z}_2;c)M^3_{13}(\bm{y}_1,\bm{y}_2;b)M^3_{12}(\bm{x}_1,\bm{x}_2;a).
$$

The above implies the system of equations
\begin{align*}
    & \bm{u}_1\bm{v}_1=\bm{y}_1\bm{x}_1,\quad c\bm{u}_1\bm{v}_2\bm{w}_2^{-1}+\bm{u}_2\bm{w}_1=\bm{y}_1\bm{x}_2,\quad \bm{u}_2\bm{w}_2=\bm{y}_2,\\
    &a\bm{u}_2^{-1}\bm{v}_1=a\bm{z}_1\bm{x}_2^{-1}+b\bm{z}_2\bm{y}_2^{-1}\bm{x}_1,\quad ac\bm{u}_2^{-1}\bm{v}_2\bm{w}_2^{-1}=b\bm{z}_2\bm{y}_2^{-1}\bm{x}_2,\\ 
    & b\bm{v}_2^{-1}=ca\bm{z}_2^{-1}\bm{x}_2^{-1}.
\end{align*}
The above system can be solved to give the following correspondence
\begin{align*}
    &\bm{u}_2=a\bm{u}_1^{-1}\bm{y}_1\bm{x}_1 (a\bm{z}_1\bm{x}_2^{-1}+b\bm{z}_2\bm{y}_2^{-1}\bm{x}_1)^{-1},\quad \bm{v}_1=\bm{u}_1^{-1}\bm{y}_1\bm{x}_1,\quad \bm{v}_2=a^{-1}bc^{-1}\bm{x}_2\bm{z}_2,\\
    &\bm{w}_1=a^{-1}(a\bm{z}_1\bm{x}_2^{-1}+b\bm{z}_2\bm{y}_2^{-1}\bm{x}_1)\bm{x}_1^{-1}\bm{y}_1^{-1}\bm{u}_1\left[\bm{y}_1\bm{x}_2-b\bm{y}_1\bm{x}_1(a\bm{z}_1\bm{x}_2^{-1}+b\bm{z}_2\bm{y}_2^{-1}\bm{x}_1)^{-1}\bm{z}_2\bm{y}_2^{-1}\bm{x}_2\right],
\end{align*}
which is a solution to the local Yang--Baxter equation for any $\bm{u}_1\in\mathfrak{R}$.

For the choice $\bm{u}_1=\bm{y}_1$ the above correspondence defines the map
\begin{align}\label{fully-noncomm-NLS}
    &\bm{u}_1=\bm{y}_1,\nonumber\\
    &\bm{u}_2=a\bm{x}_1 (a\bm{z}_1\bm{x}_2^{-1}+b\bm{z}_2\bm{y}_2^{-1}\bm{x}_1)^{-1},\nonumber\\ 
    &\bm{v}_1=\bm{x}_1,\\ 
    &\bm{v}_2=a^{-1}bc^{-1}\bm{x}_2\bm{z}_2,\nonumber\\
    &\bm{w}_1=a^{-1}(a\bm{z}_1\bm{x}_2^{-1}+b\bm{z}_2\bm{y}_2^{-1}\bm{x}_1)\bm{x}_1^{-1}\left[\bm{y}_1\bm{x}_2-b\bm{y}_1\bm{x}_1(a\bm{z}_1\bm{x}_2^{-1}+b\bm{z}_2\bm{y}_2^{-1}\bm{x}_1)^{-1}\bm{z}_2\bm{y}_2^{-1}\bm{x}_2\right],\nonumber\\
    &\bm{w}_2=a^{-1}(a\bm{z}_1\bm{x}_2^{-1}+b\bm{z}_2\bm{y}_2^{-1}\bm{x}_1)\bm{x}_1^{-1}\bm{y}_2,\nonumber
\end{align}
which is a noncommutative version of the map (29) in \cite{Sokor-2020}; nevertheless, it does not share the same invariants with (29) due to noncommutativity. However, if one makes the centrality assumption (see results on noncommutative Yang--Baxter maps where the centrality assumption was employed \cite{Doliwa-2014, Kassotakis-Kouloukas}), then the above map has three invariants in seperable variables.

In particular, we have the following.

\begin{proposition}
Let $\bm{x}_1,\bm{y}_1,\bm{z}_1\in Z(\mathfrak{R})$. Then, \eqref{fully-noncomm-NLS} implies the following map
\begin{subequations}\label{noncomm-NLS}
\begin{align}
    \bm{x}_1\mapsto \bm{u}_1&=\bm{y}_1,\\
    \bm{x}_2\mapsto \bm{u}_2&=a\bm{x}_1 (a\bm{z}_1\bm{x}_2^{-1}+b\bm{z}_2\bm{y}_2^{-1}\bm{x}_1)^{-1},\\ 
    \bm{y}_1\mapsto \bm{v}_1&=\bm{x}_1,\\ 
    \bm{y}_2\mapsto\bm{v}_2&=a^{-1}bc^{-1}\bm{x}_2\bm{z}_2,\\
    \bm{z}_1\mapsto\bm{w}_1&=\bm{x}_1^{-1}\bm{y}_1\bm{z}_1,\\
    \bm{z}_2\mapsto\bm{w}_2&=a^{-1}(a\bm{z}_1\bm{x}_2^{-1}+b\bm{z}_2\bm{y}_2^{-1}\bm{x}_1)\bm{x}_1^{-1}\bm{y}_2.
\end{align}
\end{subequations}
Moreover, map \eqref{noncomm-NLS} admits the following invariants
$$
I_1=\bm{x}_1\bm{y}_1,\quad I_2=\bm{x}_1+\bm{y}_1,\quad I_3=\bm{y}_1\bm{z}_1.
$$
Finally, map \eqref{noncomm-NLS} is a noncommutative parametric tetrahedron map.
\end{proposition}

\begin{proof}
If $\bm{x}_1,\bm{y}_1,\bm{z}_1\in Z(\mathfrak{R})$, then $$\bm{w}_1=a^{-1}(a\bm{z}_1\bm{x}_2^{-1}+b\bm{z}_2\bm{y}_2^{-1}\bm{x}_1)\bm{x}_1^{-1}\left[\bm{y}_1\bm{x}_2-b\bm{y}_1\bm{x}_1(a\bm{z}_1\bm{x}_2^{-1}+b\bm{z}_2\bm{y}_2^{-1}\bm{x}_1)^{-1}\bm{z}_2\bm{y}_2^{-1}\bm{x}_2\right]=\bm{x}_1^{-1}\bm{y}_1\bm{z}_1.
$$
Moreover, $\bm{u}_1\bm{v}_1=\bm{y}_1\bm{x}_1=\bm{x}_1\bm{y}_1$, $\bm{u}_1+\bm{v}_1=\bm{x}_1+\bm{y}_1$ and $\bm{v}_1\bm{w}_1=\bm{y}_1\bm{z}_1$.

Concerning the tetrahedron property, we first define the $4\times 4$ extensions of matrix ${\rm M}(\bm{x}_1,\bm{x}_2;a)$ in \eqref{DT-NLS}, i.e.:
\begin{align*}
        &{\rm M}^4_{12}=\begin{pmatrix}\bm{x}_1 & \bm{x}_2 & 0 & 0 \\ a\bm{x}_2^{-1} & 0 & 0 & 0 \\ 0 & 0 & 1 & 0 \\ 0 & 0 & 0 & 1\end{pmatrix}, \quad
    {\rm M}^4_{13}=\begin{pmatrix}\bm{x}_1 & 0 & \bm{x}_2 & 0 \\ 0 & 1 & 0 & 0 \\ a\bm{x}_2^{-1} & 0 & 0 & 0 \\ 0 & 0 & 0 & 1\end{pmatrix},\quad
        {\rm M}^4_{23}=\begin{pmatrix}1 & 0 & 0 & 0 \\ 0 & \bm{x}_1 & \bm{x}_2 & 0 \\ 0 & a\bm{x}_2^{-1} & 0 & 0 \\ 0 & 0 & 0 & 1\end{pmatrix}\nonumber\\
   & {\rm M}^4_{14}=\begin{pmatrix}\bm{x}_1 & 0 & 0 & \bm{x}_2 \\ 0 & 1 & 0 & 0 \\ 0 & 0 & 1 & 0 \\ a\bm{x}_2^{-1} & 0 & 0 & 0\end{pmatrix}, \quad
        {\rm M}^4_{24}=\begin{pmatrix}1 & 0 & 0 & 0 \\ 0 & \bm{x}_1 & 0 & \bm{x}_2 \\ 0 & 0 & 1 & 0 \\ 0 & a\bm{x}_2^{-1} & 0 & 0\end{pmatrix}, \quad
    {\rm M}^4_{34}=\begin{pmatrix}1 & 0 & 0 & 0 \\ 0 & 1 & 0 & 0 \\ 0 & 0 & \bm{x}_1 & \bm{x}_2 \\ 0 & 0 & a\bm{x}_2^{-1} & 0\end{pmatrix}.
\end{align*}

Using the left-hand side of the tetrahedron equation, and taking into account \eqref{noncomm-NLS}, we have:
\begin{align}\label{6fac-NLS-left}
   &{\rm M}^4_{12}(\bm{x}_1,\bm{x}_2,a_1){\rm M}^4_{13}(\bm{y}_1,\bm{y}_2,a_2) {\rm M}^4_{23}(\bm{z}_1,\bm{z}_2,a_3){\rm M}^4_{14}(\bm{r}_1,\bm{r}_2,a_4) {\rm M}^4_{24}(\bm{s}_1,\bm{s}_2,a_5) {\rm M}^4_{34}(\bm{t}_1,\bm{t}_2,a_6)=\nonumber\\
    &{\rm M}^4_{23}(\bm{x}_1^{-1}\bm{y}_1\bm{z}_1,\tilde{\bm{z}}_2,a_3){\rm M}^4_{13}(\bm{x}_1,\tilde{\bm{y}}_2,a_2) {\rm M}^4_{12}(\bm{y}_1,\tilde{\bm{x}}_2,a_1){\rm M}^4_{14}(\bm{r}_1,\bm{r}_2,a_4) {\rm M}^4_{24}(\bm{s}_1,\bm{s}_2,a_5) {\rm M}^4_{34}(\bm{t}_1,\bm{t}_2,a_6)=\nonumber\\
    &{\rm M}^4_{23}(\bm{x}_1^{-1}\bm{y}_1\bm{z}_1,\tilde{\bm{z}}_2,a_3){\rm M}^4_{13}(\bm{x}_1,\tilde{\bm{y}}_2,a_2) {\rm M}^4_{24}(\tilde{\bm{s}}_1,\tilde{\bm{s}}_2,a_5) {\rm M}^4_{14}(\bm{y}_1,\tilde{\bm{r}}_2,a_4) {\rm M}^4_{12}(\bm{r}_1,\tilde{\tilde{\bm{x}}}_2,a_1)  {\rm M}^4_{34}(\bm{t}_1,\bm{t}_2,a_6)=\nonumber\\
     &{\rm M}^4_{23}(\bm{x}_1^{-1}\bm{y}_1\bm{z}_1,\tilde{\bm{z}}_2,a_3){\rm M}^4_{24}(\tilde{\bm{s}}_1,\tilde{\bm{s}}_2,a_5) {\rm M}^4_{13}(\bm{x}_1,\tilde{\bm{y}}_2,a_2) {\rm M}^4_{14}(\bm{y}_1,\tilde{\bm{r}}_2,a_4) {\rm M}^4_{34}(\bm{t}_1,\bm{t}_2,a_6)  {\rm M}^4_{12}(\bm{r}_1,\tilde{\tilde{\bm{x}}}_2,a_1)=\nonumber\\
      &{\rm M}^4_{23}(\bm{x}_1^{-1}\bm{y}_1\bm{z}_1,\tilde{\bm{z}}_2,a_3){\rm M}^4_{24}(\tilde{\bm{s}}_1,\tilde{\bm{s}}_2,a_5) {\rm M}^4_{34}(\tilde{\bm{t}}_1,\tilde{\bm{t}}_2,a_6) {\rm M}^4_{14}(\bm{x}_1,\tilde{\tilde{\bm{r}}}_2,a_4) {\rm M}^4_{13}(\bm{y}_1,\tilde{\tilde{\bm{y}}}_2,a_2)  {\rm M}^4_{12}(\bm{r}_1,\tilde{\tilde{\bm{x}}}_2,a_1)=\nonumber\\
      &{\rm M}^4_{34}(\tilde{\tilde{\bm{t}}}_1,\tilde{\tilde{\bm{t}}}_2,a_6){\rm M}^4_{24}(\bm{x}_1^{-1}\bm{y}_1\bm{z}_1,\tilde{\tilde{\bm{s}}}_2,a_5){\rm M}^4_{23}(\tilde{\bm{s}}_1,\tilde{\tilde{\bm{z}}}_2,a_3)  {\rm M}^4_{14}(\bm{x}_1,\tilde{\tilde{\bm{r}}}_2,a_4) {\rm M}^4_{13}(\bm{y}_1,\tilde{\tilde{\bm{y}}}_2,a_2)  {\rm M}^4_{12}(\bm{r}_1,\tilde{\tilde{\bm{x}}}_2,a_1)=\nonumber\\
      &{\rm M}^4_{34}(\tilde{\tilde{\bm{t}}}_1,\tilde{\tilde{\bm{t}}}_2,a_6){\rm M}^4_{24}(\bm{x}_1^{-1}\bm{y}_1\bm{z}_1,\tilde{\tilde{\bm{s}}}_2,a_5) {\rm M}^4_{14}(\bm{x}_1,\tilde{\tilde{\bm{r}}}_2,a_4) {\rm M}^4_{23}(\tilde{\bm{s}}_1,\tilde{\tilde{\bm{z}}}_2,a_3) {\rm M}^4_{13}(\bm{y}_1,\tilde{\tilde{\bm{y}}}_2,a_2)  {\rm M}^4_{12}(\bm{r}_1,\tilde{\tilde{\bm{x}}}_2,a_1).
\end{align}
Furthermore, using the right-hand side of the tetrahedron equation, and taking into account \eqref{noncomm-NLS}, we have:
\begin{align}\label{6fac-NLS-right}
   &{\rm M}^4_{12}(\bm{x}_1,\bm{x}_2,a_1){\rm M}^4_{13}(\bm{y}_1,\bm{y}_2,a_2) {\rm M}^4_{23}(\bm{z}_1,\bm{z}_2,a_3){\rm M}^4_{14}(\bm{r}_1,\bm{r}_2,a_4) {\rm M}^4_{24}(\bm{s}_1,\bm{s}_2,a_5) {\rm M}^4_{34}(\bm{t}_1,\bm{t}_2,a_6)=\nonumber\\
   &{\rm M}^4_{12}(\bm{x}_1,\bm{x}_2,a_1){\rm M}^4_{13}(\bm{y}_1,\bm{y}_2,a_2) {\rm M}^4_{14}(\bm{r}_1,\bm{r}_2,a_4) {\rm M}^4_{23}(\bm{z}_1,\bm{z}_2,a_3) {\rm M}^4_{24}(\bm{s}_1,\bm{s}_2,a_5) {\rm M}^4_{34}(\bm{t}_1,\bm{t}_2,a_6)=\nonumber\\
   &{\rm M}^4_{12}(\bm{x}_1,\bm{x}_2,a_1){\rm M}^4_{13}(\bm{y}_1,\bm{y}_2,a_2) {\rm M}^4_{14}(\bm{r}_1,\bm{r}_2,a_4) {\rm M}^4_{34}(\hat{\bm{t}}_1,\hat{\bm{t}}_2,a_6) {\rm M}^4_{24}(\bm{z}_1,\hat{\bm{s}}_2,a_5) {\rm M}^4_{23}(\bm{s}_1,\hat{\bm{z}}_2,a_3) =\nonumber
   \\
   &{\rm M}^4_{12}(\bm{x}_1,\bm{x}_2,a_1) {\rm M}^4_{34}(\hat{\hat{\bm{t}}}_1,\hat{\hat{\bm{t}}}_2,a_6)  {\rm M}^4_{14}(\bm{y}_1,\hat{\bm{r}}_2,a_4) {\rm M}^4_{13}(\bm{r}_1,\hat{\bm{y}}_2,a_2) {\rm M}^4_{24}(\bm{z}_1,\hat{\bm{s}}_2,a_5) {\rm M}^4_{23}(\bm{s}_1,\hat{\bm{z}}_2,a_3) =\nonumber \\
   &{\rm M}^4_{34}(\hat{\hat{\bm{t}}}_1,\hat{\hat{\bm{t}}}_2,a_6) {\rm M}^4_{12}(\bm{x}_1,\bm{x}_2,a_1)  {\rm M}^4_{14}(\bm{y}_1,\hat{\bm{r}}_2,a_4) {\rm M}^4_{24}(\bm{z}_1,\hat{\bm{s}}_2,a_5) {\rm M}^4_{13}(\bm{r}_1,\hat{\bm{y}}_2,a_2)  {\rm M}^4_{23}(\bm{s}_1,\hat{\bm{z}}_2,a_3) =\nonumber\\
   &{\rm M}^4_{34}(\hat{\hat{\bm{t}}}_1,\hat{\hat{\bm{t}}}_2,a_6) {\rm M}^4_{24}(\bm{x}_1^{-1}\bm{y}_1\bm{z}_1,\hat{\hat{\bm{s}}}_2,a_5) {\rm M}^4_{14}(\bm{x}_1,\hat{\hat{\bm{r}}}_2,a_4) {\rm M}^4_{12}(\bm{y}_1,\hat{\bm{x}}_2,a_1)   {\rm M}^4_{13}(\bm{r}_1,\hat{\bm{y}}_2,a_2)  {\rm M}^4_{23}(\bm{s}_1,\hat{\bm{z}}_2,a_3) =\nonumber\\
   &{\rm M}^4_{34}(\hat{\hat{\bm{t}}}_1,\hat{\hat{\bm{t}}}_2,a_6) {\rm M}^4_{24}(\bm{x}_1^{-1}\bm{y}_1\bm{z}_1,\hat{\hat{\bm{s}}}_2,a_5) {\rm M}^4_{14}(\bm{x}_1,\hat{\hat{\bm{r}}}_2,a_4)  {\rm M}^4_{23}(\hat{\bm{s}}_1,\hat{\hat{\bm{z}}}_2,a_3) {\rm M}^4_{13}(\bm{y}_1,\hat{\hat{\bm{y}}}_2,a_2)  {\rm M}^4_{12}(\bm{r}_1,\hat{\hat{\bm{x}}}_2,a_1).
      \end{align}
      From \eqref{6fac-NLS-left} and \eqref{6fac-NLS-right} follows that
    \begin{align*}
        &{\rm M}^4_{34}(\tilde{\tilde{\bm{t}}}_1,\tilde{\tilde{\bm{t}}}_2,a_6){\rm M}^4_{24}(\bm{x}_1^{-1}\bm{y}_1\bm{z}_1,\tilde{\tilde{\bm{s}}}_2,a_5) {\rm M}^4_{14}(\bm{x}_1,\tilde{\tilde{\bm{r}}}_2,a_4) {\rm M}^4_{23}(\tilde{\bm{s}}_1,\tilde{\tilde{\bm{z}}}_2,a_3) {\rm M}^4_{13}(\bm{y}_1,\tilde{\tilde{\bm{y}}}_2,a_2)  {\rm M}^4_{12}(\bm{r}_1,\tilde{\tilde{\bm{x}}}_2,a_1)=\\
        &{\rm M}^4_{34}(\hat{\hat{\bm{t}}}_1,\hat{\hat{\bm{t}}}_2,a_6) {\rm M}^4_{24}(\bm{x}_1^{-1}\bm{y}_1\bm{z}_1,\hat{\hat{\bm{s}}}_2,a_5) {\rm M}^4_{14}(\bm{x}_1,\hat{\hat{\bm{r}}}_2,a_4)  {\rm M}^4_{23}(\hat{\bm{s}}_1,\hat{\hat{\bm{z}}}_2,a_3) {\rm M}^4_{13}(\bm{y}_1,\hat{\hat{\bm{y}}}_2,a_2)  {\rm M}^4_{12}(\bm{r}_1,\hat{\hat{\bm{x}}}_2,a_1).
    \end{align*}
    
The above implies 
\begin{align*}
&\tilde{\tilde{\bm{s}}}_2\tilde{\tilde{\bm{r}}}_2^{-1}\tilde{\tilde{\bm{y}}}_2=\hat{\hat{\bm{s}}}_2\hat{\hat{\bm{r}}}_2^{-1}\hat{\hat{\bm{y}}}_2,\quad \tilde{\tilde{\bm{t}}}_2^{-1}\tilde{\tilde{\bm{z}}}_2\tilde{\tilde{\bm{y}}}_2^{-1}\bm{r}_1 =\hat{\hat{\bm{s}}}_2\hat{\hat{\bm{r}}}_2^{-1}\hat{\hat{\bm{y}}}_2,\\
&a_1\bm{x}_1^{-1}\bm{y}_1\bm{z}_1
\tilde{\bm{s}}_1\tilde{\tilde{\bm{x}}}_2^{-1}+a_4\tilde{\tilde{\bm{s}}}_2\tilde{\tilde{\bm{r}}}_2^{-1}\bm{y}_1\bm{r}_1=a_1\bm{x}_1^{-1}\bm{y}_1\bm{z}_1
\hat{\bm{s}}_1\hat{\hat{\bm{x}}}_2^{-1}+a_4\hat{\hat{\bm{s}}}_2\hat{\hat{\bm{r}}}_2^{-1}\bm{y}_1\bm{r}_1,\\
& \tilde{\tilde{\bm{t}}}_2\tilde{\tilde{\bm{s}}}_2^{-1}\tilde{\tilde{\bm{z}}}_2\tilde{\tilde{\bm{y}}}_2^{-1}\tilde{\tilde{\bm{x}}}_2 =\hat{\hat{\bm{t}}}_2\hat{\hat{\bm{s}}}_2^{-1}\hat{\hat{\bm{z}}}_2\hat{\hat{\bm{y}}}_2^{-1}\hat{\hat{\bm{x}}}_2,\quad \bm{x}_1^{-1}\bm{y}_1\bm{z}_1\tilde{\tilde{\bm{z}}}_2\tilde{\tilde{\bm{y}}}_2^{-1}\tilde{\tilde{\bm{x}}}_2=\bm{x}_1^{-1}\bm{y}_1\bm{z}_1\hat{\hat{\bm{z}}}_2\hat{\hat{\bm{y}}}_2^{-1}\hat{\hat{\bm{x}}}_2,\\
&a_2
\tilde{\tilde{\bm{t}}}_1\tilde{\tilde{\bm{z}}}_2\tilde{\tilde{\bm{y}}}_2^{-1}\bm{r}_1+a_1a_5\tilde{\tilde{\bm{t}}}_2\tilde{\tilde{\bm{s}}}_2^{-1}\tilde{\bm{s}}_1\tilde{\tilde{\bm{x}}}_2^{-1}=a_2
\hat{\hat{\bm{t}}}_1\hat{\hat{\bm{z}}}_2\hat{\hat{\bm{y}}}_2^{-1}\bm{r}_1+a_1a_5\hat{\hat{\bm{t}}}_2\hat{\hat{\bm{s}}}_2^{-1}\hat{\bm{s}}_1\hat{\hat{\bm{x}}}_2^{-1},
\end{align*}
and $\tilde{\tilde{\bm{r}}}_2=\hat{\hat{\bm{r}}}_2$, $\bm{x}_1\bm{y}_1\tilde{\tilde{\bm{x}}}_2=\bm{x}_1\bm{y}_1\hat{\hat{\bm{x}}}_2$ and $\bm{x}_1\tilde{\tilde{\bm{y}}}_2=\bm{x}_1\hat{\hat{\bm{y}}}_2$. From the latter follows that $\tilde{\tilde{\bm{r}}}_2=\hat{\hat{\bm{r}}}_2$, $\tilde{\tilde{\bm{x}}}_2=\hat{\hat{\bm{x}}}_2$ and $\tilde{\tilde{\bm{y}}}_2=\hat{\hat{\bm{y}}}_2$ in view of which the above system implies that $\tilde{\bm{s}}_1=\hat{\bm{s}}_1$, $\tilde{\tilde{\bm{s}}}_2=\hat{\hat{\bm{s}}}_2$, $\tilde{\tilde{\bm{t}}}_1=\hat{\hat{\bm{t}}}_1$, $\tilde{\tilde{\bm{z}}}_2=\hat{\hat{\bm{z}}}_2$ and $\tilde{\tilde{\bm{t}}}_2=\hat{\hat{\bm{t}}}_2$. Therefore, map \eqref{noncomm-NLS} is a noncommutative parametric tetrahedron map.
\end{proof}

\subsection{Restriction to a noncommutative tetrahedron Sergeev's map}
In this section, we use the invariants and we restrict map \eqref{noncomm-NLS} to invariant leaves. 

In particular, we have the following. 

\begin{theorem}
Map \eqref{noncomm-NLS} can be restricted to the following three-dimensional map $(\bm{x}_2,\bm{y}_2,\bm{z}_2)\rightarrow (\bm{u}_2,\bm{v}_2,\bm{w}_2)$ given by
\begin{subequations}\label{noncomm-Sergeev-restr}
\begin{align}
    &\bm{u}_2=a(a\bm{x}_2^{-1}+b\bm{z}_2\bm{y}_2^{-1})^{-1},\\
    &\bm{v}_2=a^{-1}bc^{-1}\bm{x}_2\bm{z}_2,\\
    &\bm{w}_2=a^{-1}(a\bm{x}_2^{-1}+b\bm{z}_2\bm{y}_2^{-1})y_2.
\end{align}
\end{subequations}
This three-parametric family of maps includes the fully noncommutative Sergeev's map:
\begin{subequations}\label{noncomm-Sergeev}
\begin{align}
    \bm{x}\mapsto\bm{u}&=(\bm{x}^{-1}+\bm{z}\bm{y}^{-1})^{-1},\\
    \bm{y}\mapsto\bm{v}&=k^{-1}\bm{x}\bm{z},\\
    \bm{z}\mapsto\bm{w}&=\bm{x}^{-1}\bm{y}+\bm{z}.
\end{align}
\end{subequations}

The noncommutative Sergeev's map \eqref{noncomm-Sergeev} has Lax representation
\begin{equation}\label{noncomm-Sergeev-Lax}
    {\rm L}^3_{12}(\bm{u};k){\rm L}^3_{13}(\bm{v};k){\rm L}^3_{23}(\bm{w};k)= {\rm L}^3_{23}(\bm{z};k){\rm L}^3_{13}(\bm{y};k){\rm L}^3_{12}(\bm{x};k),
\end{equation}
where ${\rm L}(\bm{x};k)=\begin{pmatrix}
1 & \bm{x} \\
k\bm{x}^{-1} & 0
\end{pmatrix}$. Moreover, it is a noncommutative, noninvolutive tetrahedron map.
\end{theorem}

\begin{proof}
Map \eqref{noncomm-Sergeev} admits the invariants $I_1=\bm{x}_1\bm{y}_1$, $I_2=\bm{x}_1+\bm{y}_1$, $I_3=\bm{y}_1\bm{z}_1$. Setting, $I_1=I_3=1$, $I_2=2$ it follows that $\bm{x}_1=\bm{y}_1=\bm{z}_1$. Substituting to \eqref{noncomm-Sergeev}, we obtain the three-dimensional map \eqref{noncomm-Sergeev-restr}.

Now, choosing $a=b=c=k\in Z(\mathfrak{R})$ in \eqref{noncomm-Sergeev-restr}, we obtain the fully noncommutative map \eqref{noncomm-Sergeev}. This is the noncommutative version of Sergeev's map, namely, if we assume that $\bm{x},\bm{y},\bm{z}\in Z(\mathfrak{R})$, then map \eqref{noncomm-Sergeev-restr} becomes map (20) in \cite{Sergeev}. Moreover, by substitution of ${\rm L}(\bm{x};k)=\begin{pmatrix}
1 & \bm{x} \\
k\bm{x}^{-1} & 0
\end{pmatrix}$ to \eqref{noncomm-Sergeev-Lax}, map \eqref{noncomm-Sergeev-restr} follows by straightforward calculations.

Regarding the tetrahedron property, we consider the $4\times 4$ generalisations of matrix ${\rm L}(\bm{x};k)$, namely the following
\begin{align*}
    &{\rm L}^4_{12}(\bm{x};k)=\begin{pmatrix}1 & \bm{x} & 0 & 0 \\ k\bm{x}^{-1} & 0 & 0 & 0 \\ 0 & 0 & 1 & 0 \\ 0 & 0 & 0 & 1\end{pmatrix}, \quad
    {\rm L}^4_{13}(\bm{x};k)=\begin{pmatrix}1 & 0 & \bm{x} & 0 \\ 0 & 1 & 0 & 0 \\ k\bm{x}^{-1} & 0 & 0 & 0 \\ 0 & 0 & 0 & 1\end{pmatrix},\quad
    {\rm L}^4_{23}(\bm{x};k)=\begin{pmatrix}1 & 0 & 0 & 0 \\ 0 & 1 & \bm{x} & 0 \\ 0 & k\bm{x}^{-1} & 0 & 0 \\ 0 & 0 & 0 & 1\end{pmatrix}\nonumber\\
   & {\rm L}^4_{14(\bm{x};k)}=\begin{pmatrix}1 & 0 & 0 & \bm{x} \\ 0 & 1 & 0 & 0 \\ 0 & 0 & 1 & 0 \\ k\bm{x}^{-1} & 0 & 0 & 0\end{pmatrix}, \quad
    {\rm L}^4_{24(\bm{x};k)}=\begin{pmatrix}1 & 0 & 0 & 0 \\ 0 & 1 & 0 & \bm{x} \\ 0 & 0 & 1 & 0 \\ 0 & k\bm{x}^{-1} & 0 & 0\end{pmatrix}, \quad
    {\rm L}^4_{34}(\bm{x};k)=\begin{pmatrix}1 & 0 & 0 & 0 \\ 0 & 1 & 0 & 0 \\ 0 & 0 & 1 & \bm{x} \\ 0 & 0 & k\bm{x}^{-1} & 0\end{pmatrix}.
\end{align*}
Substitution of the above to \eqref{6-fac} for $a_i=k$, $i=1,\ldots 6$, implies
\begin{align*}
    &\hat{\bm{x}}^{-1}+\hat{\bm{z}}\hat{\bm{y}}^{-1}+\hat{\bm{s}}\hat{\bm{r}}^{-1}=\bm{x}^{-1}+\bm{z}\bm{y}^{-1}+\bm{s}\bm{r}^{-1},\\
    &(\hat{\bm{z}}\hat{\bm{y}}^{-1}+\hat{\bm{s}}\hat{\bm{r}}^{-1})\hat{\bm{x}}=(\bm{z}\bm{y}^{-1}+\bm{s}\bm{r}^{-1})\bm{x},\quad \hat{\bm{s}}\hat{\bm{r}}^{-1}\hat{\bm{y}}=\bm{s}\bm{r}^{-1}\bm{y},\\
    & \hat{\bm{t}}\hat{\bm{s}}^{-1}\hat{\bm{z}}\hat{\bm{y}}^{-1}\hat{\bm{x}}=\bm{t}\bm{s}^{-1}\bm{z}\bm{y}^{-1}\bm{x},\quad \hat{\bm{t}}^{-1}\hat{\bm{z}}^{-1}\hat{\bm{x}}^{-1}=\bm{t}^{-1}\bm{z}^{-1}\bm{x}^{-1},\\
    &\hat{\bm{z}}^{-1}\hat{\bm{x}}^{-1}+\hat{\bm{t}}\hat{\bm{s}}^{-1}\hat{\bm{x}}^{-1}+\hat{\bm{t}}\hat{\bm{s}}^{-1}\hat{\bm{z}}\hat{\bm{y}}^{-1}=\bm{z}^{-1}\bm{x}^{-1}+\bm{t}\bm{s}^{-1}\bm{x}^{-1}+\bm{t}\bm{s}^{-1}\bm{z}\bm{y}^{-1},
\end{align*}
as well as $\hat{\bm{x}}=\bm{x}$, $\hat{\bm{y}}=\bm{y}$ and $\hat{\bm{r}}=\bm{r}$ in view of which the above system implies $\hat{\bm{s}}=\bm{s}$, $\hat{\bm{z}}=\bm{z}$ and $\hat{\bm{t}}=\bm{t}$. Therefore, map \eqref{noncomm-Sergeev} satisfies the tetrahedron equation.

Finally, the noninvolutivity follows, for instance, from $\bm{v}\circ (\bm{u},\bm{v},\bm{w})=k^{-1}(\bm{x}^{-1}+\bm{z}\bm{y}^{-1})^{-1}(\bm{x}^{-1}\bm{y}+\bm{z})\neq \bm{y}$.
\end{proof}

\section{Conclusions}\label{conclusions}
In this paper, we showed which additional matrix six-factorisation condition must be satisfied for a map with Lax representation to be a tetrahedron map, namely we proved Theorem \ref{six-factorisation}. A similar result was obtained in \cite{Kouloukas-Papageorgiou} for the case of Yang--Baxter maps.

Moreover, we derived new noninvolutive, noncommutative tetrahedron maps, namely maps \eqref{gauge-mKdV-map} and \eqref{1d-rel-map}, generated via the local Yang--Baxter equation by employing: i. a matrix related to a gauge transformation for the lattice mKdV equation \cite{Frank-Walker}, and ii. a generalised matrix related to one-dimensional relativistic elastic collisions of two particles \cite{Kouloukas}. Maps \eqref{gauge-mKdV-map} and \eqref{1d-rel-map} were used as illustrative examples to show that the matrix six-factorisation condition also works for correspondences.

Finally, we showed how to construct noncommutative tetrahedron maps using Darboux transformations of nonocommutative integrable systems. In particular, using a Darboux transformation for the noncommutative coupled NLS system, we constructed a noncommutative tetrahedron map \eqref{noncomm-NLS}. We showed that the latter can be restricted to a noncommutative Sergeev's map, namely map \eqref{noncomm-Sergeev}.

All the tetrahedron maps derived in this paper are noninvolutive. Noninvolutive maps are more interesting in terms of their dynamics.

The results can be extended in the following ways.

\begin{itemize}
    \item In \cite{IKKRP} the algebraic properties of linear tetrahedron maps were studied, and linearised versions of nonlinear tetrahedron maps were considered by linearising their associated Lax representations. One could study the properties of the linearised versions of nonlinear, noncommutive tetrahedron maps derived in this paper.
    
    \item In \cite{Sergei-Sokor} a matrix trifactorisation problem other from the local Yang--Baxter equation was used to derive tetrahedron maps. One can can employ this matrix trifactorisation problem in order to derive new, noncommutative solutions to Zamolodchikov's tetrahedron equation. Moreover, it is expected that a theorem similar to \ref{six-factorisation} will hold for maps being generated by this matrix trifactorisation problem.
    
    \item The Liouville integrability of all the derived tetrahedron maps in this paper is an open problem. It is expected that centrality assumptions should be made.
    
    \item Study the associated noncommutative $3D$ lattice systems. There several methods in the literature on how to associate tetrahedron maps to lattice equations (for example, using symmetries \cite{Kassotakis-Tetrahedron}). Moreover, well-known methods for deriving $2D$ lattice equations from Yang--Baxter maps could be also employed in order to associate noncommutative tetrahedron maps with noncommutative lattice equations (for instance, as lifts to the corresponding tetrahedron map \cite{Kouloukas-Tran, Pap-Tongas} or using the inviariants of the maps in separable form \cite{Pavlos-Maciej, Pavlos-Maciej-2}.) It makes sense to compare these B\''acklund type of transformations to those derived using the method presented in \cite{FKRX} for integrable $2D$ lattice equations.
    
    \item The noncommutative difference equations associated with the tetrahedron maps derived in section \ref{NLS-noncomm_tetrahedron_maps} are probably certain discretisations of noncommutative NLS type of PDEs. One may construct these nonlinear systems of PDEs by taking certain continuum limits. It is expected that the tetrahedron maps of section \ref{NLS-noncomm_tetrahedron_maps} will be B\"acklund type of transformations for these nonlinear systems of PDEs.
\end{itemize}

\section{Acknowledgements}
The work on sections \ref{prelim}, \ref{Matrix-six-factorisation_problem} and \ref{conclusions} was funded by the Russian Science Foundation (project No. 21-71-30011), whereas the work on sections \ref{intro} and \ref{NLS-noncomm_tetrahedron_maps}  was carried out within the framework of a development programme for the Regional Scientific and Educational Mathematical Centre of the P.G. Demidov Yaroslavl State University
with financial support from the Ministry of Science and Higher Education of the Russian Federation
(Agreement on provision of subsidy from the federal budget No. 075-02-2022-886).

I would like to thank Dr. S. Igonin and Ms. A. Kutuzova for several useful discussions.


\begin{thebibliography}{10}
\bibitem{Baxter-1983}
{R.J. Baxter,} {On Zamolodchikov's Solution of the Tetrahedron Equations,} {Commun. Math. Phys. 88 (1983)} {185--205}.

\bibitem{Baxter-1986}
{R.J. Baxter,} {The Yang--Baxter equations and the Zamolodchikov model,} {Phys. D: Nonlinear Phenom. 18 (1986)} {321--347}.

\bibitem{Bazhanov-Sergeev}
{V.V. Bazhanov and S.M. Sergeev,} {Zamolodchikov's tetrahedron equation and hidden structure of quantum groups,} {J.
Phys. A 39 (2006)} {3295--3310}.

\bibitem{Bazhanov-Mangazeev-Sergeev}
{V.V. Bazhanov, V.V. Mangazeev, and S.M. Sergeev,} {Quantum geometry of three-dimensional lattices,} {J. Stat. Mech.
(2008)} {P07004}.

\bibitem{Bobenko-Suris}
{A.I. Bobenko and Yu.B. Suris,} {Integrable noncommutative equations
on quad-graphs. The consistency approach,} {Lett. Math. Phys.} {61 (2002)} {241--254}.

\bibitem{Dimakis-Hoissen}
{A. Dimakis and F. M\"uller-Hoissen,} {Burgers and Kadomtsev-Petviashvili hierarchies: A functional representation approach,} {Theor. Math. Phys. 152 (2007)} {933--947}.

\bibitem{Dimakis-Hoissen-2015}
{A. Dimakis and F. M\"uller-Hoissen,} {Simplex and polygon equations,} {SIGMA 11 (2015)} {042}.


\bibitem{Doliwa-2014}
{A. Doliwa,} {Non-Commutative Rational Yang--Baxter Maps,} {Lett. Math. Phys.} {104 (2014)} {299--309}.

\bibitem{Doliwa-Kashaev}
{A. Doliwa and R.M. Kashaev,} {Non-commutative bi-rational maps satisfying Zamolodchikov equation equation, and Desargues lattices,} {J. Math. Phys.} {61 (2020)} {092704}.

\bibitem{Doliwa-Noumi}
{A. Doliwa and M. Noumi,} {The Coxeter relations and KP map for non-commuting symbols,} {Lett. Math. Phys. 110 (2020)} {2743--2762}.


\bibitem{FKRX}
{X. Fisenko, S. Konstantinou-Rizos, and P. Xenitidis.} 
{A discrete Darboux-Lax scheme for integrable difference equations} 
{\em Chaos, Solitons and Fractals}  {\textbf{158}} {112059} (2022).

\bibitem{IKKRP}
{S. Igonin, V. Kolesov, S. Konstantinou-Rizos, and M. Preobrazhenskaia,} 
{Tetrahedron maps, Yang--Baxter maps, and partial linearisations} 
{\em J. Phys. A: Math. Theor.} {\textbf{54}} {505203} (2021).

\bibitem{Sergei-Sokor}
{S. Igonin and S. Konstantinou-Rizos,} 
{Algebraic and differential-geometric constructions of set-theoretical solutions to the Zamolodchikov tetrahedron equation} 
{\em arXiv:2110.05998} (2021).

\bibitem{Kassotakis-Kouloukas}
{P. Kassotakis and T. Kouloukas,}
{On non-abelian quadrirational Yang--Baxter maps.} 
{\em arXiv:2109.11975 } (2021)

\bibitem{Pavlos-Maciej}
{P. Kassotakis and M. Nieszporski,} {Families of integrable equations,} {SIGMA 7 (2011) 100} {14pp}.

\bibitem{Pavlos-Maciej-2}
{P. Kassotakis and M. Nieszporski,} {On non-multiaffine consistent-around-the-cube lattice equations,} {Phys. Lett. A 376 (2012)} {3135--3140}.

\bibitem{Kassotakis-Tetrahedron}
{P. Kassotakis, M. Nieszporski, V. Papageorgiou, and A. Tongas,} {Tetrahedron maps and symmetries of three dimensional integrable discrete equations,} {J. Math. Phys. 60 (2019)} {123503}.

\bibitem{Kapranov-Voevodsky} 
{M.M. Kapranov and V.A. Voevodsky,} {2-categories and Zamolodchikov tetrahedra equations. In: Algebraic Groups and
Their Generalizations: Quantum and Infinite-Dimensional Methods (University Park, PA, 1991), pp. 177--259,} {Proc.
Sympos. Pure Math. 56 Amer. Math. Soc., Providence, RI, 1994.}

\bibitem{Kashaev-Sergeev} 
{R.M. Kashaev, I.G. Koperanov, and S.M. Sergeev,} {Functional Tetrahedron Equation,} {Theor. Math. Phys.} {117 (1998)} {370--384}.


\bibitem{Sokor-2020}
{S. Konstantinou-Rizos,}
{Nonlinear Schr\"odinger type tetrahedron maps.} 
{\em Nuclear Phys. B } {\textbf{960}} (2020) {115207} 

\bibitem{Sokor-Xenitidis}
{S. Konstantinou-Rizos and P. Xenitidis,} 
{Integrable discretisations of a noncommutative NLS equation.} 
{\em (in preparation)} (2022).

\bibitem{Korepanov}
{I.G. Korepanov,} {Algebraic integrable dynamical systems, $2+1$-dimensional models in wholly discrete space-time, and inhomogeneous models in 2-dimensional statistical physics,} {(1995) solv-int/9506003.}

\bibitem{Kouloukas}
{T.E. Kouloukas,} {Relativistic collisions as Yang--Baxter maps,} {Phys. Lett. A 381 (2017)} {3445--3449}.


\bibitem{Kouloukas-Papageorgiou}
{T.E. Kouloukas and V.G. Papageorgiou,} {Yang--Baxter maps with first-degree polynomial $2\times 2$ Lax matrices,} {J. Phys. A: Math. Theor. 42 (2009)} {404012}.

\bibitem{Kouloukas-Tran}
{T. E. Kouloukas and D. Tran,} {Poisson structures for lifts and periodic reductions of integrable lattice equations,} {J. Phys. A: Math. Theor. 48 (2015)} {075202}.


\bibitem{Kupershmidt}
{B.A. Kupershmidt,} {Noncommutative mathematics of Lagrangian, Hamiltonian, and integrable systems,} {Math. surv. and monographs 78 (2000), AMS,} {http://dx.doi.org/10.1090/surv/078}.

\bibitem{Nijhoff-Capel}
{F.W. Nijhoff and H.W. Capel,} {The direct linearization approach to hierarchies of integrable PDE's in $2+1$ dimensions. I. Lattice equations and the differential-difference Hierarchies,} {Inverse Problems 6 (1990)} {567--590}.

\bibitem{Nijhoff}
{J.M. Maillet, F. Nijhoff,} {The tetrahedron equation and the four-simplex equation,} {Phys. Lett. A 134 (1989)} {221--228}.

\bibitem{Maillet-Nijhoff}
{J.M. Maillet, F. Nijhoff,} {Integrability for multidimensional lattice models,} {Phys. Lett. B 224 (1989)} {389--396}.


\bibitem{Frank-Walker}
{F.W. Nijhoff and A.J. Walker,} {The discrete and continuous Painlev\'e VI hierarchy and the Garnier systems,} {Glasgow Mathematical Journal 43 (2001)} {109--123}.

\bibitem{Nimmo}
{J.J.C. Nimmo,} {On a non-Abelian Hirota-Miwa equation,} {J. Phys. A: Math. Gen. 39 (2006)} {5053}.

\bibitem{Pap-Tongas}
{V.G. Papageorgiou and A.G. Tongas,} {Yang--Baxter maps associated to elliptic curves,} {(2009) arXiv:0906.3258v1.}

\bibitem{Sergeev}
{S.M. Sergeev,} {Solutions of the Functional Tetrahedron Equation Connected with the Local Yang--Baxter Equation for the Ferro-Electric Condition,} {Lett. Math. Phys. 45 (1998)}  {113--119}.

\bibitem{Sokolov}
{V. Sokolov,} 
{Algebraic structures in integrability} 
{Vladimir Sokolov. ``Algebraic Structures in Integrability''} {\em World Scientific Publishing Co. Pte. Ltd, Singapore} (2020). DOI: 10.1142/11809 

\bibitem{Talalaev}
{D.V. Talalaev,} 
{Tetrahedron equation: algebra, topology, and integrability} {Russian Math. Surveys 76 (2021)} {685--721}.

\bibitem{Zamolodchikov} 
{A.B. Zamolodchikov,} {Tetrahedra equations and integrable systems in three-dimensional space,} {Sov. Phys. JETP 52 (1980)} {325--336}.

\bibitem{Zamolodchikov-2} 
{A.B. Zamolodchikov,} {Tetrahedron equations and the relativistic S matrix of straight strings in (2+1)-dimensions,} {Commun. Math. Phys. 79 (1981)} {489--505}.

\end{thebibliography}
\end{document}